\documentclass{article}

\usepackage{pgfplots}
\usepackage[letterpaper, margin=1in]{geometry}
\usepackage{gensymb}
\usepackage{url,amssymb,amsmath,amsthm,amscd,paralist,bbm,wrapfig}
\usepackage{tkz-graph}
\usepackage{tikz-cd}
\usepackage{amsmath, bm}
\usepackage{fancyhdr}
\usepackage{amssymb}
\usepackage{graphicx}
\usepackage{algorithm}
\usepackage{algorithmic} 

\usepackage[T1]{fontenc}
\usepackage[numbers]{natbib}
\tikzstyle{vertex}=[circle,black, fill=black, draw, inner sep=0pt, minimum size=6pt]

\usepackage{tikz}
\usetikzlibrary{decorations.pathreplacing}
\usepackage{xcolor}
\usepackage{mathtools}
\definecolor{cof}{RGB}{219,144,71}
\definecolor{pur}{RGB}{186,146,162}
\definecolor{greeo}{RGB}{91,173,69}
\definecolor{greet}{RGB}{52,111,72}
\pgfplotsset{compat=1.14}
\begin{document}
\newcommand{\footremember}[2]{%
   \footnote{#2}
    \newcounter{#1}
    \setcounter{#1}{\value{footnote}}%
}
\newcommand{\footrecall}[1]{%
    \footnotemark[\value{#1}]%
} 

\newcommand{\mbf}[1]{\ensuremath\mathbf{#1}}
\newcommand{\R}{\mathbb{R}}
\newcommand{\C}{\mathbb{C}}
\newcommand{\al}{\alpha}
\newcommand{\PiWdix}{\Pi_{i=d}^1 W_{i,+,x}}
\newcommand{\PiWdiy}{\Pi_{i=d}^1 W_{i,+,y}}
\newcommand{\PiWdixo}{\Pi_{i=d}^1 W_{i,+,x_0}}
\newcommand{\inner}[2]{\langle #1,#2 \rangle}
\newcommand{\relu}{\text{relu}}
\newcommand{\E}{\operatorname{\mathbb{E}}}
\newcommand{\Pro}{\operatorname{\mathbb{P}}}
\newcommand{\one}{\operatorname{\mathbbm{1}}}
\newcommand{\sgn}{\operatorname{sgn}}
\newcommand{\sign}{\operatorname{sgn}}
\newcommand{\diag}{\operatorname{diag}}
\newcommand{\dimension}{\operatorname{dim}}
\newcommand{\argmin}{\operatorname{argmin}}

\newcommand{\G}{\mathcal{G}}
\newcommand{\MM}{\mathcal{M}}
\newcommand{\F}{\mathcal{F}} 
\newcommand{\SG}{\mathcal{S}(G)}
 \newcommand{\SZG}{\mathcal{S}_0(G)}
\newcommand{\So}{\mathcal{S}_0}
\newcommand{\Sl}{\mathcal{S}_\ell}
\newcommand{\I}{\mathcal{I}}
\newcommand{\K}{\mathcal{K}}
\newcommand{\mr}{\operatorname{mr}}
\newcommand{\mro}{\operatorname{mr}_0}
\newcommand{\MRo}{\operatorname{MR}_0}
\newcommand{\mrl}{\operatorname{mr}_\ell}
\newcommand{\mrp}{\operatorname{mr}_+}
\newcommand{\mrs}{\operatorname{mr}_-}
\newcommand{\hmr}{\operatorname{hmr}}
\newcommand{\mrF}{\operatorname{mr}^F}
\newcommand{\mrZ}{\operatorname{mr}^{\ZZ}}
\newcommand{\M}{\operatorname{M}}
\newcommand{\Mo}{\operatorname{M}_0}
\newcommand{\Mp}{\operatorname{M}_+}
\newcommand{\Zo}{\operatorname{Z}_0}
\newcommand{\perm}{\operatorname{perm}}
\newcommand{\cyc}{\operatorname{cyc}}
\newcommand{\nc}{\operatorname{nc}}
\newcommand{\nev}{\operatorname{ne}}
\newcommand{\sgc}{\operatorname{sgc}}
\newcommand{\pma}{\operatorname{pm}}
\newcommand{\match}{\operatorname{match}}
\newcommand{\nul}{\operatorname{null}}
\newcommand{\rank}{\operatorname{rank}}
\newcommand{\rk}{\operatorname{rk}^{\Z_2}}
\newcommand{\w}{\operatorname{\omega}}
\newcommand{\cov}{\operatorname{cov}}
\newcommand{\var}{\operatorname{Var}}

\newtheorem{thm}{Theorem}
\newtheorem{lem}{Lemma}
\newtheorem{cor}[thm]{Corollary}
\newtheorem{prop}{Proposition}
\newtheorem{prob}[thm]{Problem}
\newtheorem{alg}[thm]{Algorithm}
\newtheorem{nota}[thm]{Notation}

\theoremstyle{definition}
\newtheorem{defn}[thm]{Definition}
\newtheorem{ex}{Example}

\theoremstyle{remark}
\newtheorem{rem}[thm]{Honor Pledge}
\newtheorem{obs}[thm]{Observation} 
\author{Paul Hand\footremember{Rice}{Department of Computational and Applied Mathematics, Rice University, Houston, TX}, Oscar Leong\footrecall{Rice} , and Vladislav Voroninski\footremember{Helm}{Helm.ai, Menlo Park, CA}}
\title{Phase Retrieval Under a Generative Prior}
\maketitle

\begin{abstract}
The phase retrieval problem asks to recover a natural signal $y_0 \in \mathbb{R}^n$ from $m$ quadratic observations, where $m$ is to be minimized. As is common in many imaging problems, natural signals are considered sparse with respect to a known basis, and the generic sparsity prior is enforced via $\ell_1$ regularization. While successful in the realm of linear inverse problems, such $\ell_1$ methods have encountered possibly fundamental limitations, as no computationally efficient algorithm for phase retrieval of a $k$-sparse signal has been proven to succeed with fewer than $O(k^2\log n)$ generic measurements, exceeding the theoretical optimum of $O(k \log n)$. In this paper, we propose a novel framework for phase retrieval by 1) modeling natural signals as being in the range of a deep generative neural network $G : \mathbb{R}^k \rightarrow \mathbb{R}^n$  and 2) enforcing this prior directly by optimizing an empirical risk objective over the domain of the generator. Our formulation has provably favorable global geometry for gradient methods, as soon as $m = O(kd^2\log n)$, where $d$ is the depth of the network. Specifically, when suitable deterministic conditions on the generator and measurement matrix are met, we construct a descent direction for any point outside of a small neighborhood around the unique global minimizer and its negative multiple, and show that such conditions hold with high probability under Gaussian ensembles of multilayer fully-connected generator networks and measurement matrices. This formulation for structured phase retrieval thus has two advantages over sparsity based methods: 1) deep generative priors can more tightly represent natural signals and 2) information theoretically optimal sample complexity. We corroborate these results with experiments showing that exploiting generative models in phase retrieval tasks outperforms sparse phase retrieval methods.

\end{abstract}

\section{Introduction}

We study the problem of recovering a signal $y_0 \in \R^n$ given $m \ll n$ phaseless observations of the form $b = |Ay_0|$ where the measurement matrix $A \in \R^{m \times n}$ is known and $|\cdot|$ is understood to act entrywise. This is known as the \textit{phase retrieval problem}. In this work, we assume, as a prior, that the signal $y_0$ is in the range of a generative model $G : \R^k \rightarrow \R^n$ so that $y_0 = G(x_0)$ for some $x_0 \in \R^k$. To recover $y_0$, we first recover the original latent code $x_0$ corresponding to it, from which $y_0$ is obtained by applying $G$. Hence we study the \textit{phase retrieval problem under a generative prior}: \[
\text{find}\ x \in \R^k\ \text{such that}\ b = |AG(x)|.
\] 

We will refer to this formulation as Deep Phase Retrieval (DPR). The phase retrieval problem has applications in X-ray crystallography \cite{Harrison93,Millane90}, optics \cite{Walther1963}, astronomical imaging \cite{Fienup87}, diffraction imaging \cite{diffImaging}, and microscopy \cite{MicroscopyPR}. In these problems, the phase information of an object is lost due to physical limitations of scientific instruments. In crystallography, the linear measurements in practice are typically Fourier modes because they are the far field limit of a diffraction pattern created by emitting a quasi-monochromatic wave on the object of interest.

In many applications, the signals to be recovered are compressible or sparse with respect to a certain basis (e.g. wavelets). Many researchers have attempted to leverage sparsity priors in phase retrieval to yield more efficient recovery algorithms 
\cite{Cai2016, Voroninski2013, SPARTA, Jain2013,SparsePRoverview,Hedge2017}. However, these methods have been met with potentially severe fundamental limitations. In the Gaussian measurement regime where $A$ has i.i.d. Gaussian entries, one would hope that recovery of a $k$-sparse $n$-dimensional signal is possible with $O(k\log n)$ measurements. However, there is no known method to succeed with fewer than $O(k^2 \log n)$ measurements. Moreover, \cite{Voroninski2013} proved that the semidefinite program PhaseLift cannot outperform this suboptimal sample complexity by direct $\ell_1$ penalization. This is in stark contrast to the success of leveraging sparsity in linear compressed sensing to yield optimal sample complexity. 

\paragraph{Our contribution.} We establish information theoretically optimal sample complexity\footnote{with respect to the dimensionality of the latent code given to the generative network} for structured phase retrieval under generic measurements and a novel nonlinear formulation based on empirical risk under a generative prior. In this work, we suppose that the signal of interest is in the range of a generative model. In particular, the generative model is a $d$-layer, fully-connected, feed forward neural network with Rectifying Linear Unit (ReLU) activation functions and no bias terms. Let $W_i \in \R^{n_i \times n_{i-1}}$ denote the weights in the $i$-th layer of our network for $i = 1,\dots,d$ where $k = n_0 < n_1 < \dots < n_d$. Given an input $x \in \R^k$, the output of the the generative model $G : \R^k \rightarrow \R^{n_d}$ can be expressed as  \begin{align*}
G(x) : = \relu\left(W_d \dots\relu(W_2(\relu(W_1x)))\dots\right)
\end{align*} where $\relu(x) = \max(x,0)$ acts entrywise. To recover $x_0$ from the measurements $|AG(x_0)|$, we study the following $\ell_2$ empirical risk minimization problem: \begin{align} \min_{x \in \R^k} f(x) := \frac{1}{2}\left\||AG(x)| - |AG(x_0)|\right\|^2. \label{objective} \end{align} We establish a set of deterministic conditions on the generator $G$ and a measurement matrix $A$, which guarantee that the empirical risk formulation has favorable global geometry for gradient methods. We then establish that these deterministic conditions hold with high probability when $G$ has random Gaussian weights and $A$ is a Gaussian measurement matrix. Our deterministic conditions likely accommodate other probability distributions on the generator weights and measurement matrices. Moreover, we note that while we assume the weights to have i.i.d. Gaussian entries, we make no assumption about the independence between layers. The Gaussian assumption of the weight matrices is supported by empirical evidence showing neural networks, learned from data, that have weights that obey statistics similar to Gaussians \cite{GaussNets}. Furthermore, there has also been work done in establishing a relationship between deep networks and Gaussian processes \cite{Lee2018}. 

Due to the non-convexity of (\ref{objective}), there is no a-priori guarantee that gradient descent schemes can solve (\ref{objective}) as many local minima may exist. In spite of this, our main result illustrates that the objective function exhibits favorable geometry for gradient methods. Moreover, our result holds with information theoretically optimal sample complexity:

\begin{thm}[Informal]
If we have a sufficient number of measurements $m = \Omega(kd \log(n_1\dots n_d))$ and our network is sufficiently expansive at each layer $n_i = \Omega(n_{i-1}\log n_{i-1})$, then there exists a descent direction $v_{x,x_0} \in \R^k$ for any non-zero $x \in \R^k$ outside of two small  neighborhoods centered at the true solution $x_0$ and a negative multiple $-\rho_d x_0$ with high probability. In addition, the origin is a local maximum of $f$. Here $\rho_d > 0$ depends on the number of layers $d$ and $\rho_d \rightarrow 1$ as $d \rightarrow \infty$.
\end{thm}

Our main result asserts that the objective function does not have any spurious local minima or saddle points away from neighborhoods of the true solution and a negative multiple of it. Hence if one were to solve (\ref{objective}) via gradient descent and the algorithm converged, the final iterate would be close to the true solution or a negative multiple of it. The proof of this result is a concentration argument. We first prove the sufficiency of two deterministic conditions on the weights $W_i$ and measurement matrix $A$. We then show that Gaussian $W_i$ and $A$ satisfy these conditions with high probability. Finally, using these two conditions, we argue that the specified descent direction $v_{x,x_0}$ concentrates around a vector $h_{x,x_0}$ that is continuous for non-zero $x \in \R^k$ and vanishes only when $x \approx x_0$ or $x \approx -\rho_d x_0$.

Based on our main theoretical result, we also propose a gradient descent scheme to solve (\ref{objective}) in Section 2. We solved phase retrieval tasks on both synthetic and natural signals in Section 4, comparing our algorithm's results to three sparse phase retrieval methods. Our results corroborate our main theoretical result that recovery with $O(k d^2 \log n)$ measurements is possible while outperforming the alternative sparse phase retrieval methods. In particular, our proposed framework is advantageous in allowing enforcement of tighter priors than sparsity, and has information theoretically optimal sample complexity, unlike $\ell_1$-based methods.



\paragraph{Prior methodologies for general phase retrieval.} In the Gaussian measurement regime, most of the techniques to solve phase retrieval problems can be classified as convex or non-convex methods. In terms of convex techniques, lifting-based methods transform the signal recovery problem into a rank-one matrix recovery problem by \textit{lifting} the signal into the space of positive semi-definite matrices. These semidefinite programming (SDP) approaches, such as Phaselift \cite{CandStroh2013}, can be provably recover any $n$-dimensional signal with $O(n \log n)$ measurements. Other convex methods include PhaseCut \cite{Phasecut}, an SDP approach, and linear programming algorithms such as PhaseMax \cite{Phasemax}

Non-convex methods encompass alternating minimization approaches such as the original Gerchberg-Saxton \cite{GerchSaxAlg} and Fienup \cite{FienupAlg} algorithms and direct optimization algorithms such as Wirtinger Flow \cite{Candes2015}. These latter methods directly tackle the least squares objective function: \begin{align}
\min_{y \in \R^n} \frac{1}{2} \left\| |Ay|^2 - |Ay_0|^2 \right\|^2. \label{wirtflowfunc}
\end{align} In the seminal work, \cite{Candes2015} show that through an initialization via the spectral method, a gradient descent scheme can solve (\ref{wirtflowfunc}) where the gradient is understood in the sense of Wirtinger calculus with $O(n \log n)$ measurements. Expanding on this, a later study on the minimization of (\ref{wirtflowfunc}) in \cite{Wright2016} showed that with $O(n \log^3 n)$ measurements, the energy landscape of the objective function exhibited global benign geometry which would allow it to be solved efficiently by gradient descent schemes without special initialization. There also exist amplitude flow methods that solve the following non-smooth variation of (\ref{wirtflowfunc}):
\begin{align}
\min_{y \in \R^n} \frac{1}{2} \left\| |Ay| - |Ay_0| \right\|^2. \label{ampflowfunc}
\end{align} These methods have found success with only $O(n)$ measurements \cite{Wang2017} and have shown to empirically perform better than intensity-based methods using the squared formulation in (\ref{wirtflowfunc}) \cite{Yeh2015}.

\paragraph{Sparse phase retrieval.} Many of the successful methodologies for general phase retrieval have been adapted to try to solve sparse phase retrieval problems. In terms of non-convex optimization, wirtinger flow type methods such as Thresholded Wirtinger Flow \cite{Cai2016} create a sparse initializer via the spectral method and perform thresholded gradient descent updates to generate sparse iterates to solve (\ref{wirtflowfunc}). Another non-convex method, SPARTA \cite{SPARTA}, estimates the support of the signal for its initialization and performs hard thresholded gradient updates to the amplitude-based objective function (\ref{ampflowfunc}). Both of these methods require $O(k^2 \log n)$ measurements for a generic $k$-sparse $n$-dimensional signal to succeed, which is more than the theoretical optimum $O(k \log n)$.

While lifting-based methods such as Phaselift have been proven unable to beat the suboptimal sample complexity $O(k^2 \log n)$, there has been some progress towards breaking this barrier. In \cite{HV2016}, the authors show that with an initializer that sufficiently correlates with the true solution, a linear program can recover the sparse signal from $O(k \log \frac{n}{k})$ measurements. However, the best known initialization methods require at least $O(k^2 \log n)$ measurements \cite{Cai2016}. Outside of the Gaussian measurement regime, there have been other results showing that if one were able to design their own measurement matrices, then the optimal sample complexity could be reached \cite{SparsePRoverview}. For example, \cite{Romberg2015} showed that assuming the measurement vectors were chosen from an incoherent subspace, then recovery is possible with $O(k \log \frac{n}{k})$ measurements. However, these results would be difficult to generalize to the experimental setting as their design architectures are often unrealistic. Moreover, the Gaussian measurement regime more closely models the experimental Fourier diffraction measurements observed in, for example, X-ray crystallography. As Fourier models are the ultimate goal, results towards lowering this sample complexity in the Gaussian measurement regime must be made or new modes of regularization must be explored in order for the theory and practice of phase retrieval to advance.

\paragraph{Related work.} There has been recent empirical evidence supporting applying a deep learning based approach to holographic imaging, a phase retrieval problem. The authors in \cite{Rivenson2017} show that a neural network with ReLU activation functions can learn to perform holographic image reconstruction. In particular, they show that compared to current approaches, this neural network based method requires less measurements to succeed and is computationally more efficient, needing only one hologram to reconstruct the necessary images. 

Furthermore, there have been a number of recent advancements in leveraging generative priors over sparsity priors in compressed sensing. In \cite{Price2017}, the authors considered the least squares objective: \begin{align}
\min_{x \in \R^k} \frac{1}{2} \|AG(x) - AG(x_0)\|^2. \label{compsensfunc}
\end{align} They provided empirical evidence showing that 5-10X fewer measurements were needed to succeed in recovery compared to standard sparsity-based approaches such as Lasso. In terms of theory, they showed that if $A$ satisfied a restricted eigenvalue condition and if one were able to solve (\ref{compsensfunc}), then the solution would be close to optimal. The authors in \cite{Hand2017} analyze the same optimization problem as in \cite{Price2017} but establish global theoretical guarantees regarding the non-convex objective function. Under particular conditions about the expansivity of each neural network layer and randomness assumptions on their weights, they show that the energy landscape of the objective function does not have any spurious local minima. Furthermore, they show that there is always a descent direction outside of two small neighborhoods of the global minimum and a negative scalar multiple of it. The success of leveraging generative priors in compressed sensing along with the sample complexity bottlenecks in sparse phase retrieval have influenced this work to consider enforcing a generative prior in phase retrieval to surpass sparse phase retrieval's current theoretical and practical limitations.

\subsection{Notation}

Let $(\cdot)^\top$ denote the real transpose. Let $[n] = \{1,\dots,n\}$. Let $\one_{S}$ denote the indicator function on the set $S$. For a vector $v \in \R^n$, $\text{diag}(v > 0)$ is $1$ in the $i$-th diagonal entry if $v_i > 0$ and $0$ otherwise. Let $\mathcal{B}(x,r)$ denote the Euclidean ball centered at $x$ with radius $r$. Let $\|\cdot\|$ denote the $\ell_2$ norm for vectors and spectral norm for matrices. For any non-zero $x \in \R^n$, let $\hat{x} = x/\|x\|$. Let $\Pi_{i=d}^1 W_i = W_d W_{d-1} \dots W_1$. Let $I_n$ be the $n \times n$ identity matrix. Let $S^{k-1}$ denote the unit sphere in $\R^k$. We write $c = \Omega(\delta)$ when $c \geq C\delta$ for some positive constant $C$. Similarly, we write $c = O(\delta)$ when $c \leq C \delta$ for some positive constant $C$. When we say that a constant depends polynomially on $\epsilon^{-1}$, this means that it is at least $C \epsilon^{-k}$ for some positive $C$ and positive integer $k$. For notational convenience, we write $a = b + O_1(\epsilon)$ if $\|a - b\| \leq \epsilon$ where $\|\cdot\|$ denotes $|\cdot|$ for scalars, $\ell_2$ norm for vectors, and spectral norm for matrices. Define $\sgn : \R \rightarrow \R$ to be $\sgn(x) = x/|x|$ for non-zero $x \in \R$ and $\sgn(x) = 0$ otherwise. For a vector $v \in \R^n$, $\text{diag}(\sgn(v))$ is $\sgn(v_i)$ in the $i$-th diagonal entry. Let $g^{\circ d}$ denote the composition of a function $g$ with itself $d$ times.

\section{Algorithm}

While our main result illustrates that the objective function exhibits favorable geometry for optimization, it does not guarantee recovery of the signal as gradient descent algorithms could, in principle, converge to the negative multiple of our true solution. Hence we propose a gradient descent scheme to recover the desired solution by escaping this region. First, consider Figure \ref{fig:objfunclandscape} which illustrates the behavior of our objective function \textit{in expectation}, i.e. when the number of measurements $m \rightarrow \infty$.

\begin{figure}[!htbp]
  \centering
  \includegraphics[scale = 0.25]{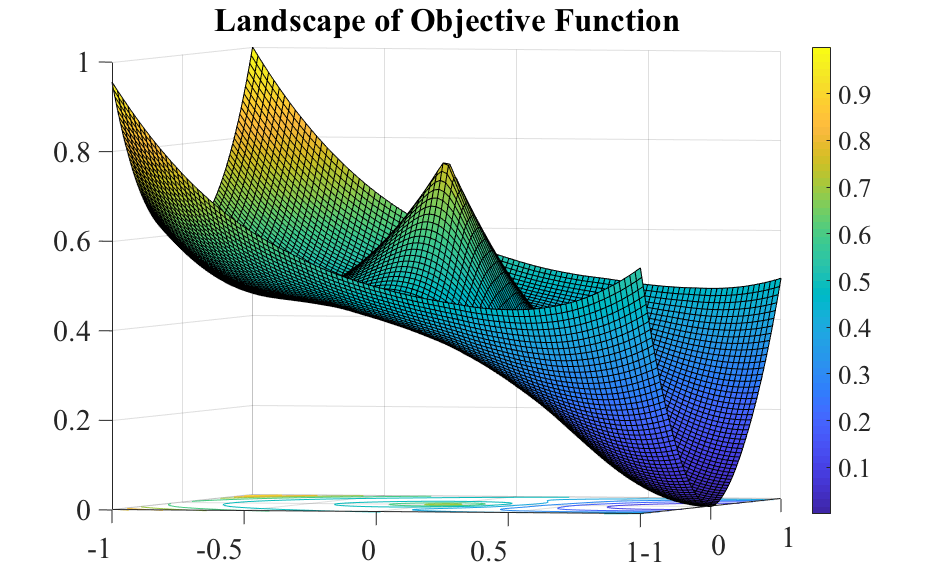}
  \caption{Energy landscape of (\ref{objective}) with $m \rightarrow \infty$ and true solution $x_0 = \left[1,\ 0\right]^\top \in \R^2$.}
  \label{fig:objfunclandscape}
\end{figure}

We observe two important attributes of the objective function's landscape: 1) there exist two minima, the true solution $x_0$ and a negative multiple $-\beta x_0$ for some $\beta > 0$ and 2) if $z \approx x_0$ while $w \approx -\beta x_0$, we have that $f(z) < f(w)$, i.e. the objective function value is lower near the true solution than near its negative multiple. This is due to the fact that the true solution is in fact the global optimum.


Based on these attributes, we will introduce a gradient descent scheme to converge to the global minimum. First, we define some useful quantities. For any $x \in \R^k$ and matrix $W \in \R^{n \times k}$, define $W_{+,x} : = \text{diag}(Wx > 0)W.$ That is, $W_{+,x}$ keeps the rows of $W$ that have a positive dot product with $x$ and zeroes out the rows that do not. We will extend the definition of $W_{+,x}$ to each layer of weights $W_i$ in our neural network. For $W_1 \in \R^{n_1 \times k}$ and $x \in \R^k$, define $W_{1,+,x} : = \text{diag}(W_1x > 0)W_1.$ For each layer $i \in [d]$, define \[W_{i,+,x} : = \text{diag}(W_i W_{i-1,+,x}\dots W_{2,+,x}W_{1,+,x}x > 0) W_i.\] $W_{i,+,x}$ keeps the rows of $W_i$ that are active when the input to the generative model is $x$. Then, for any $x \in \R^k$, the output of our generative model can be written as $G(x) = (\PiWdix)x.$ For any $z \in \R^n$, define $A_z : = \text{diag}(\sgn(Az))A.$ Note that $|AG(x)| = A_{G(x)}G(x)$ for any $x \in \R^k$.

Since a gradient descent scheme could in principle be attracted to the negative multiple, we exploit the geometry of the objective function's landscape to escape this region. First, choose a random initial iterate for gradient descent $x_1 \neq 0$. At each iteration $i = 1,2,\dots$, compute the descent direction \begin{align*}
v_{x_i,x_0} := (\Pi_{i=d}^1 W_{i,+,x_i})^\top A_{G(x_i)}^\top 
\left(|AG(x_i)| - |AG(x_0)|\right).
\end{align*} This is the gradient of our objective function $f$ where $f$ is differentiable. Once computed, we then take a step in the direction of $-v_{x_i,x_0}$. However, prior to taking this step, we compare the objective function value for $x_i$ and its negation $-x_i$. If $f(-x_i) < f(x_i)$, then we set $x_i$ to its negation, compute the descent direction and update the iterate. The intuition for this algorithm relies on the landscape illustrated in Figure \ref{fig:objfunclandscape}: since the true solution $x_0$ is the global minimum and $-\rho_d x_0$ is a local minimum, the objective function value near $x_0$ is smaller than near $-\rho_d x_0$. Hence if we begin to converge towards $-\rho_d x_0$, this algorithm will escape this region by choosing the point with lower objective function value, which will be in a neighborhood of $x_0$. Algorithm 1 formally outlines this process.

\begin{algorithm}
\caption{Deep Phase Retrieval (DPR) Gradient method}
\label{alg}
\begin{algorithmic}[1]
\REQUIRE Weights $W_i$, measurement matrix $A$, observations $|AG(x_0)|$, and step size $\alpha > 0$
\STATE Choose an arbitrary initial point $x_1 \in \mathbb{R}^k\setminus \{0\}$
\FOR {$i = 1, 2, \ldots$}  \label{alg:st3}
\IF {$f(-x_{i}) < f(x_{i})$} \label{alg:cond1}
\STATE $x_{i} \gets - x_{i}$; \label{alg:cond2}
\ENDIF
\STATE Compute $v_{x_i,x_0} =(\Pi_{i=d}^1 W_{i,+,x_i})^\top A_{G(x_i)}^\top 
\left(|AG(x_i)| - |AG(x_0)|\right)$
\STATE $x_{i+1} = x_i - \alpha  v_{x_i,x_0}$
\ENDFOR
\end{algorithmic}
\end{algorithm}

\paragraph{Remark.} We would like to note an important distinction between the landscape observed in general phase retrieval and our formulation. In general phase retrieval, there is an inherent ambiguity in the measurements $|Ay_0|$ as both $y_0$ and $-y_0$ solve (\ref{wirtflowfunc}) and (\ref{ampflowfunc}). In the complex case, there are even a continuum of solutions as $e^{i\phi}y_0$ is a solution for any $\phi \in [0,2\pi]$. In our scenario, the generative model $G$ introduces a nonlinearity that resolves this ambiguity in two ways. First, due to its injectivity: if $y_0 = G(x_0)$ for some $x_0 \in \R^k$, then $x_0$ is in fact the global optimum to (\ref{objective}). Second, as an added benefit, the non-negativity of $G$ naturally resolves the sign ambiguity since $-y_0 \notin \text{range}(G)$. While these properties certainly aid our analysis, they do not trivialize the matter as the nonlinearity of $G$ must still be handled with care.


\section{Main Theoretical Analysis}

We now formally present our main result. While the objective function is not smooth, it's one-sided directional derivatives exist everywhere due to the piecewise linearity of $G$. Let $D_vf(x)$ denote the unnormalized one-sided directional derivative of $f$ at $x$ in the direction $v$: $D_v f(x) = \lim_{t \rightarrow 0^+} \frac{f(x + tv) - f(x)}{t}.$ Note that we make no assumption about the independence between layers.

\begin{thm}
Fix $\epsilon > 0$ such that $K_1d^8 \epsilon^{1/4} \leq 1$ and let $d \geq 2$. Suppose that $G$ is such that $W_i$ has i.i.d. $\mathcal{N}(0,1/n_i)$ entries for $i = 1,\dots,d$. Suppose that $A$ has i.i.d. $\mathcal{N}(0,1/m)$ entries. If $m \geq cdk \log (n_1n_2\dots n_d)$ and $n_i \geq c n_{i-1}\log n_{i-1}$ for $i =1,\dots,d$, then with probability at least $1 - \sum_{i=1}^d \tilde{c}n_ie^{-\gamma n_{i-1}} - \tilde{\gamma}m^{4k}e^{-Cm}$, the following holds: for all non-zero $x,x_0 \in \R^k$, there exists $v_{x,x_0} \in \R^k$ such that the one-sided directional derivatives of $f$ satisfy \begin{align*}
D_{-v_{x,x_0}}f(x) < 0,\ &\ \forall x \notin \mathcal{B}(x_0, K_2 d^3 \epsilon^{1/4}\|x_0\|) \cup \mathcal{B}(-\rho_d x_0, K_2 d^{14}\epsilon^{1/4}\|x_0\|) \cup \{0\}, \\
D_x f(0) < 0,\ &\ \forall x \neq 0,
\end{align*} where $\rho_d > 0$ converges to $1$ as $d \rightarrow \infty$ and $K_1$ and $K_2$ are universal constants. Here $c$ and $\gamma^{-1}$ depend polynomially on $\epsilon^{-1}$, $C$ depends on $\epsilon$, and $\tilde{c}$ and $\tilde{\gamma}$ are universal constants.
\end{thm}


The result will be shown by proving the sufficiency of two deterministic conditions on the weights $W_i$ of our generative network and the measurement matrix $A$. 

\paragraph{Weight Distribution Condition.} The first condition quantifies the Gaussianity and spatial arrangement of the neurons in each layer. We say that $W$ satisfies the \textit{Weight Distribution Condition} (WDC) with constant $\epsilon > 0$ if for any non-zero $x,y \in \R^k$: \begin{align*}
\left\|W_{+,x}^\top W_{+,y} - Q_{x,y}\right\| \leq \epsilon\ \text{where}\ Q_{x,y}: = \frac{\pi - \theta_{x,y}}{2\pi} I_k + \frac{\sin \theta_{x,y}}{2\pi} M_{\hat{x} \leftrightarrow \hat{y}}.
\end{align*} Here $\theta_{x,y} = \angle (x,y)$ and $M_{\hat{x} \leftrightarrow \hat{y}}$ is the matrix that sends $\hat{x} \mapsto \hat{y}$, $\hat{y} \mapsto \hat{x}$, and $z \mapsto 0$ for any $z \in \text{span}(\{x,y\})^{\perp}$. If $W_{i,j} \sim \mathcal{N}(0,1/n)$, then an elementary calculation gives $\E\left[W_{+,x}^\top W_{+,y}\right] = Q_{x,y}$. \cite{Hand2017} proved that Gaussian $W$ satisfies the WDC with high probability (Lemma 1 in Appendix).


\paragraph{Range Restricted Concentration Property.}
The second condition is similar in the sense that it quantifies whether the measurement matrix behaves like a Gaussian when acting on the difference of pairs of vectors given by the output of the generative model. We say that $A \in \R^{m \times n}$ satisfies the \textit{Range Restricted Concentration Property} (RRCP) with constant $\epsilon > 0$ if for all non-zero $x,y \in \R^{k}$, the matrices $A_{G(x)}$ and $A_{G(y)}$ satisfy the following for all $x_1,x_2,x_3,x_4 \in \R^k$: \begin{align*}
|\langle (A_{G(x)}^\top A_{G(y)} - \Phi_{G(x),G(y)})(G(x_1) - G(x_2)),& G(x_3) - G(x_4)\rangle | \\
&\leq 7 \epsilon \|G(x_1) - G(x_2)\|\|G(x_3) - G(x_4)\|
\end{align*} where \begin{align*}
\Phi_{z,w} := \frac{\pi - 2 \theta_{z,w}}{\pi}I_{n} + \frac{2\sin \theta_{z,w}}{\pi}M_{\hat{z} \leftrightarrow \hat{w}}.
\end{align*} If $A_{i,j} \sim \mathcal{N}(0,1/m)$, then for any $z,w \in \R^{n}$, a similar calculation for Gaussian $W$ gives $\E\left[A_z^\top A_w\right] = \Phi_{z,w}$. In our work, we establish that Gaussian $A$ satisfies the RRCP with high probability. Please see Section 5.3 in the Appendix for a complete proof.


We emphasize that these two conditions are deterministic, meaning that our results could extend to non-Gaussian distributions. We now state our main deterministic result:

\begin{thm} Fix $\epsilon > 0$ such that $K_1d^8 \epsilon^{1/4} \leq 1$ and let $d \geq 2$. Suppose that $G$ is such that $W_i$ satisfies the WDC with constant $\epsilon$ for all $i = 1,\dots,d$. Suppose $A \in \R^{m \times n_d}$ satisfies the RRCP with constant $\epsilon$. Then the same conclusion as Theorem 2 holds.
\end{thm}

\subsection{Proof sketch for Theorem 2}

Before we outline the proof of Theorem 2, we specify the descent direction $v_{x,x_0}$. For any $x \in \R^k$ where $f$ is differentiable, we have that \begin{align*}
\nabla f(x) = (\Pi_{i=d}^1 W_{i,+,x})^\top A^\top A (\Pi_{i=d}^1 W_{i,+,x})x - (\Pi_{i=d}^1 W_{i,+,x})^\top A_{G(x)}^\top A_{G(x_0)}(\Pi_{i=d}^1 W_{i,+,x_0})x_0.
\end{align*} This is precisely the descent direction specified in DPR, expanded with our notation. When $f$ is not differentiable at $x$, choose a direction $w$ such that $f$ is differentiable at $x + \delta w$ for sufficiently small $\delta > 0$. Such a direction $w$ exists by the piecewise linearity of the generative model $G$. Hence we define our descent direction $v_{x,x_0}$ as \begin{align*}
v_{x,x_0} = \begin{cases}
\nabla f(x) & f\ \text{differentiable at } x\in \R^k \\
\lim_{\delta \rightarrow 0^+} \nabla f(x + \delta w) & \text{otherwise.}
\end{cases}
\end{align*} For non-zero $x,x_0 
\in \R^k$, let $\theta_0 = \angle (x,x_0)$. To understand how the map $x \mapsto \relu(Wx)$ distorts angles in expectation, define $g : \R \rightarrow \R$ by $$g(\theta) = \cos^{-1}\left(\frac{\cos\theta(\pi - \theta)+\sin\theta}{\pi}
\right).$$ Then for $i \geq 1$, set $\overline{\theta}_i = g(\overline{\theta}_{i-1})$ where $\overline{\theta}_0 = \theta_0$. The following is a sketch of the proof of Theorem 2:
\begin{itemize}
\item By the WDC and RRCP, we have that the descent direction $v_{x,x_0}$ concentrates around a particular vector $\overline{v}_{x,x_0}$ defined by \begin{align*}
\overline{v}_{x,x_0} & := (\PiWdix)^\top(\PiWdix)x - (\PiWdix)^\top \Phi_{G(x),G(x_0)} (\PiWdixo)x_0. 
\end{align*}
\item The WDC establishes that $\overline{v}_{x,x_0}$ concentrates uniformly for all non-zero $x,x_0 \in \R^k$ around a continuous vector $h_{x,x_0}$ defined by \begin{align*}
h_{x,x_0} & : = - \frac{\|x_0\|}{2^d}\left(\frac{\pi - 2\overline{\theta}_d}{\pi}\right)\left(\prod_{i=0}^{d-1}\frac{\pi - \overline{\theta}_i}{\pi}\right)\hat{x}_0 \\
& + \frac{1}{2^d}\left[\|x\| - \|x_0\|\left( \frac{2\sin \overline{\theta}_d}{\pi} - \left(\frac{\pi - 2\overline{\theta}_d}{\pi}\right) \sum_{i=0}^{d-1} \frac{\sin \overline{\theta}_i}{\pi}\left(\prod_{j=i+1}^{d-1} \frac{\pi - \overline{\theta}_i}{\pi}\right)\right)\right]\hat{x}.
\end{align*}
This holds uniformly for all non-zero $x,x_0 \in \R^k$.
\item A direct analysis shows that $h_{x,x_0}$ is only small in norm for $x \approx x_0$ and $x \approx -\rho_d x_0$. See section 5.2 for a complete proof. Since $v_{x,x_0} \approx \overline{v}_{x,x_0} \approx h_{x,x_0}$, $v_{x,x_0}$ is also only small in norm in neighborhoods around $x_0$ and $-\rho_d x_0$, establishing Theorem 3.
\item Gaussian $W_i$ and $A$ satisfy the WDC and RRCP with high probability (Lemma 1 and Proposition 2 in Appendix). 
\end{itemize}

Theorem 2 is a combination of Lemma 1, Proposition 2, and Theorem 3. The full proofs of these results can be found in the Appendix.

\section{Numerical Experiments}
In this section, we investigate the use of enforcing generative priors on phase retrieval tasks. We used the proposed DPR Gradient Method (Algorithm 1) and compared our results with three sparse phase retrieval algorithms: the sparse truncated amplitude flow algorithm (SPARTA) \cite{SPARTA}, Thresholded Wirtinger Flow (TWF) \cite{Cai2016}, and the alternating minimization algorithm CoPRAM \cite{Hedge2017}. For the remainder of this section, we will refer to the DPR Gradient Method as DPR.

\subsection{Experiments for Gaussian signals}

We first consider synthetic experiments using Gaussian measurements on Gaussian signals. In particular, we considered a two layer network given by $G(x) = \relu(W_2\relu(W_1x))$ where each $W_i$ has i.i.d. $\mathcal{N}(0,1)$ entries for $i = 1,2$. We set $k = 10$, $n_1 = 500$,  and $n_2 = 1000$.  We let the entries of $A \in \R^{m\times n_2}$ and $x_0 \in \R^k$ be i.i.d. 
$\mathcal{N}(0,1)$.  We ran DPR for $25$ random instances of $(A,W_1,W_2,x_0)$. A reconstruction $x^{\star}$ is considered successful if the relative error $\|x^{\star}-x_0\|/\|x_0\| \leq 10^{-4}$. We also compared our results with SPARTA, TWF, and CoPRAM. In this setting, we chose a $k=10$-sparse $y_0 \in \R^{n_2}$, where the nonzero coefficients are i.i.d. $\mathcal{N}(0,1)$.  As before, we ran each algorithm with $25$ random instances of $(A,y_0)$. Figure \ref{fig:phase_plot} displays the percentage of successful trials for different ratios $m/n$ where $n = n_2 = 1000$ and $m$ is the number of measurements. Since SPARTA was the most successful amongst the sparse methods, we also experimented with lower sparsity levels $k = 3,5$. We note that DPR achieves nearly the same empirical success rate of recovering a $10$-dimensional latent code as SPARTA in recovering a $3$-sparse $1000$-dimensional signal. For more comparable sparsity levels to the latent code dimension, DPR outperforms each alternative method with far fewer measurements.

\begin{figure}[!htb]
  \centering
  \includegraphics[scale = 0.12]{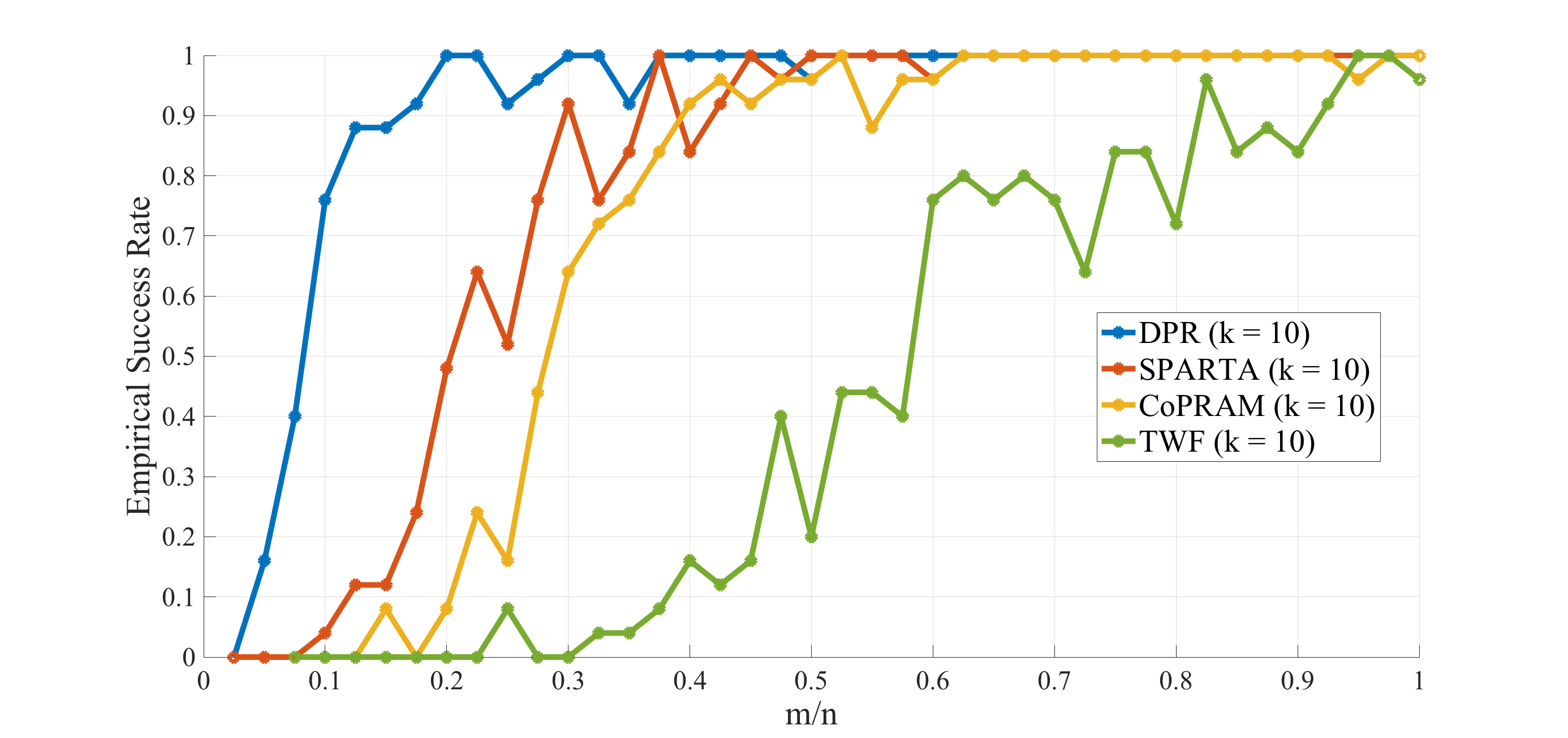}
  \includegraphics[scale = 0.12]{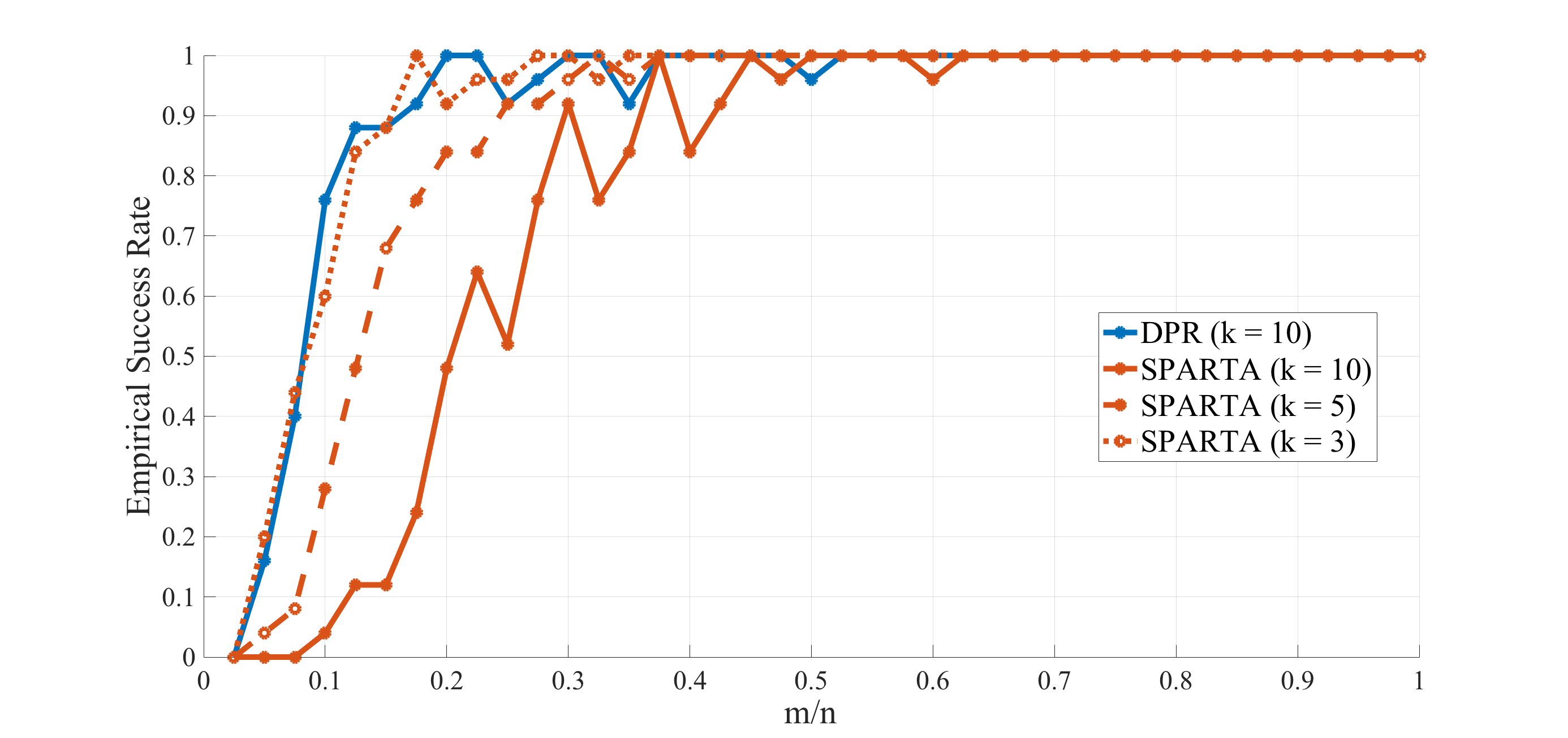}
  \caption{Empirical success rate with ratios $m/n$ where $n = 1000$ with DPR's latent code dimension set to $k = 10.$ The sparsity level for each recovery algorithm is initially set to $k = 10$ (top) while we experiment with SPARTA's sparsity levels from $k = 3,5,10$ (bottom).}
  \label{fig:phase_plot}
\end{figure}

\subsection{Experiments for MNIST and CelebA}

We next consider image recovery tasks, where we use two different generative models for the MNIST and CelebA datasets. In each task, the goal is to recover an image $y_0$ given $|Ay_0|$ where $A \in \R^{m \times n}$ has i.i.d. $\mathcal{N}(0,1/m)$ entries. We found an estimate image $G(x^{\star})$ in the range of our generator via gradient descent, using the Adam optimizer \cite{ADAM}. Empirically, we noticed that DPR would typically only negate the latent code (Lines \ref{alg:cond1}--\ref{alg:cond2}) at the initial iterate, if necessary. Hence we use a modified version of DPR in these image experiments: we ran two sessions of gradient descent for a random initial iterate $x_1$ and its negation $-x_1$ and chose the most successful reconstruction.

In the first image experiment, we used a pretrained Variational Autoencoder (VAE) from \cite{Price2017} that was trained on the MNIST dataset \cite{MNIST}. This dataset consists of $60,000$ $28 \times 28$ images of handwritten digits. As described in \cite{Price2017}, the recognition network is of size $784 - 500 - 500 - 20$ while the generator network is of size $20 - 500 - 500 - 784$. The latent code space dimension is $k = 20$.

For the sparse phase retrieval methods, we performed sparse recovery by transforming the images using the Daubechies-4 Wavelet Transform. More specifically, we considered the vector of wavelet coefficients $v_0 = \Psi y_0$ which is compressible. Since the wavelet transform $\Psi$ is orthogonal and the measurement matrix $A$ is Gaussian, $A\Psi^\top$ and $A$ are equal in distribution due to the rotational invariance of $A$. Hence, instead of recovering $v_0$ from measurements $|Ay_0| = |A\Psi^\top v_0|$, we recover $v_0$ from $|Av_0|$. We then take an inverse wavelet transform to retrieve the approximate image. To appropriately take the wavelet transform, we pad the images with zeros uniformly around the border so that they are of size $32 \times 32$. If a sparse recovery algorithm required a sparsity parameter, we ran the algorithm with a range of sparsity parameter values, choosing the best reconstruction in terms of lowest reconstruction error. The resulting images generated by our algorithm were also uniformly padded with zeros around the border to obtain $32 \times 32$ images.

\begin{figure}[h]
  \centering
  \includegraphics[scale = 0.12]{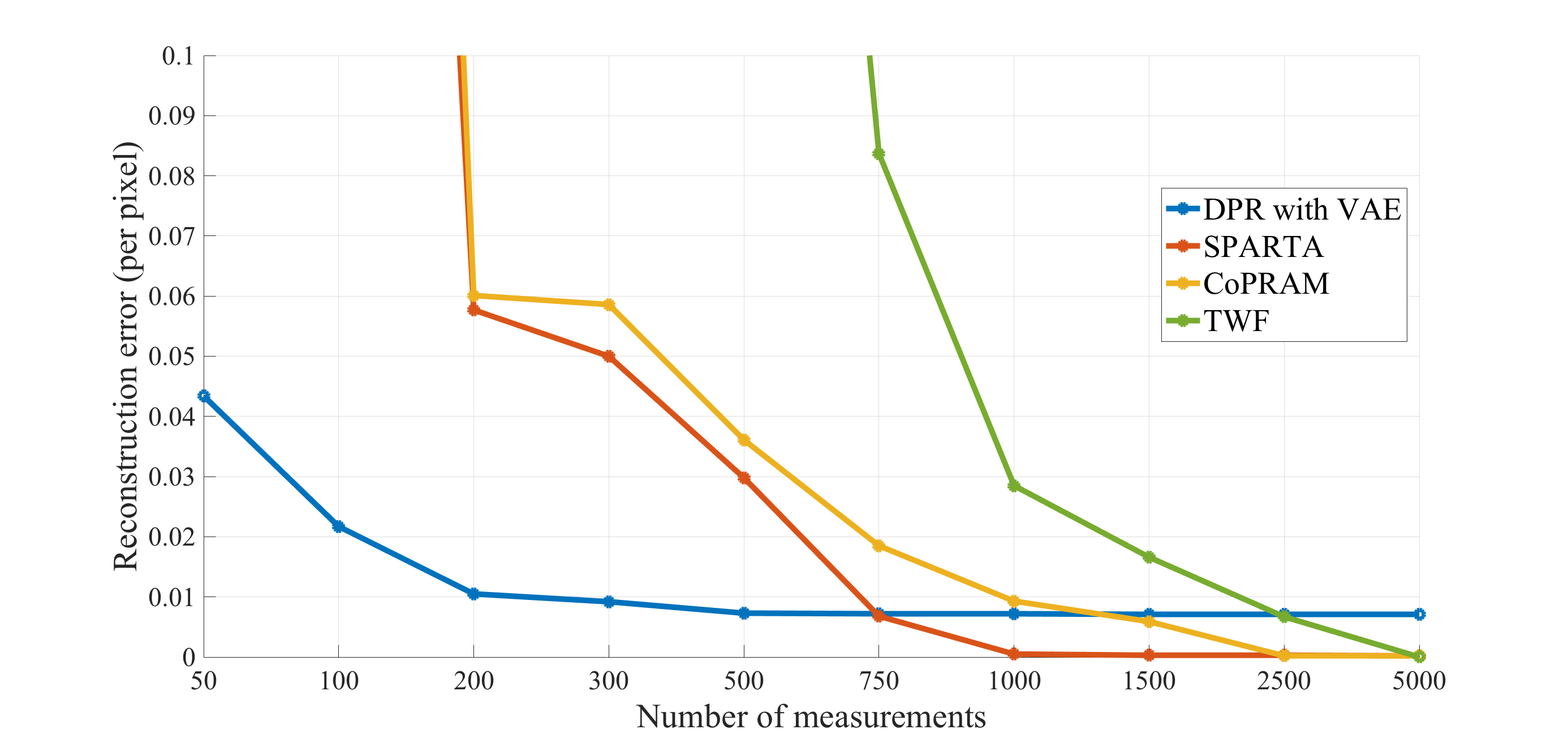}
  \includegraphics[scale = 0.12]{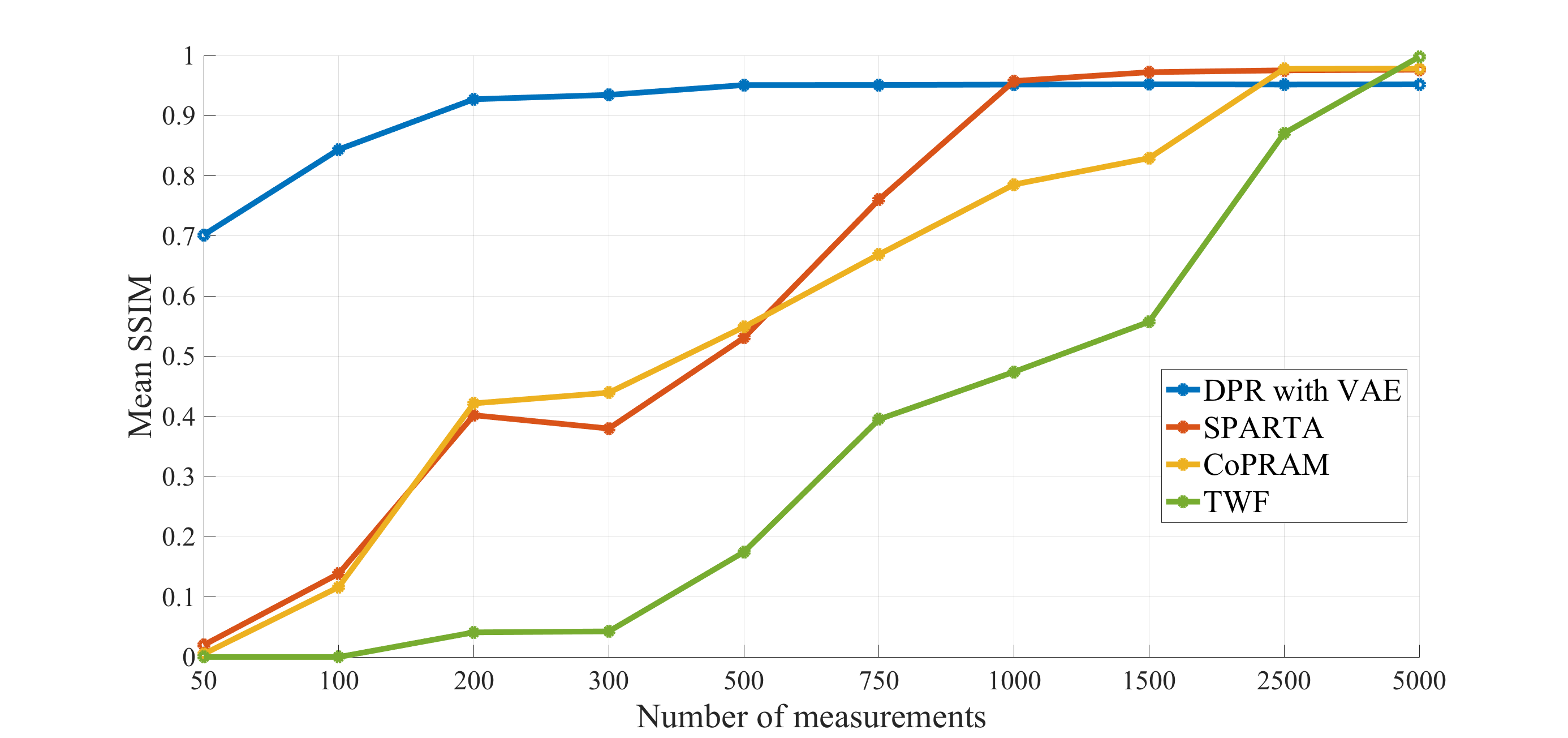}
  \caption{Each algorithm's average reconstruction error (top) and mean SSIM (bottom) over $10$ images from the MNIST test set for different numbers of measurements.}
  \label{fig:mnist_quant_results}
\end{figure}

We attempted to reconstruct $10$ images from the MNIST test set. We allowed $5$ random restarts for each algorithm and recorded the result with the least $\ell_2$ reconstruction error per pixel. In addition, we calculated the Structural Similarity Index Measure (SSIM) \cite{Wang2004} for each reconstruction and computed the average for various numbers of measurements. The results in Figure \ref{fig:mnist_quant_results} demonstrate the success of our algorithm with very few measurements. For $200$ measurements, we can achieve accurate recovery with a mean SSIM value of over $0.9$ while other algorithms require $1000$ measurements or more. In terms of reconstruction error, our algorithm exhibits recovery with $200$ measurements comparable to the alternatives requiring $750$ measurements or more, which is where they begin to succeed.

We note that while our algorithm succeeds with fewer measurements than the other methods, our performance, as measured by per-pixel reconstruction error, saturates as the number of measurements increases since our reconstruction accuracy is ultimately bounded by the generative model's representational error. As generative models improve, their representational errors will decrease.  Nonetheless, as can be seen in the reconstructed digits in Figure \ref{fig:mnist_qualitative_results}, the recoveries are semantically correct (the correct digit is legibly recovered) even though the reconstruction error does not decay to zero.  In applications, such as MRI and molecular structure estimation via X-ray crystallography, semantic error measures would be a more informative estimates of recovery performance than per-pixel error measures.

\begin{figure}[h]
  \centering
    \includegraphics[scale = 0.1]{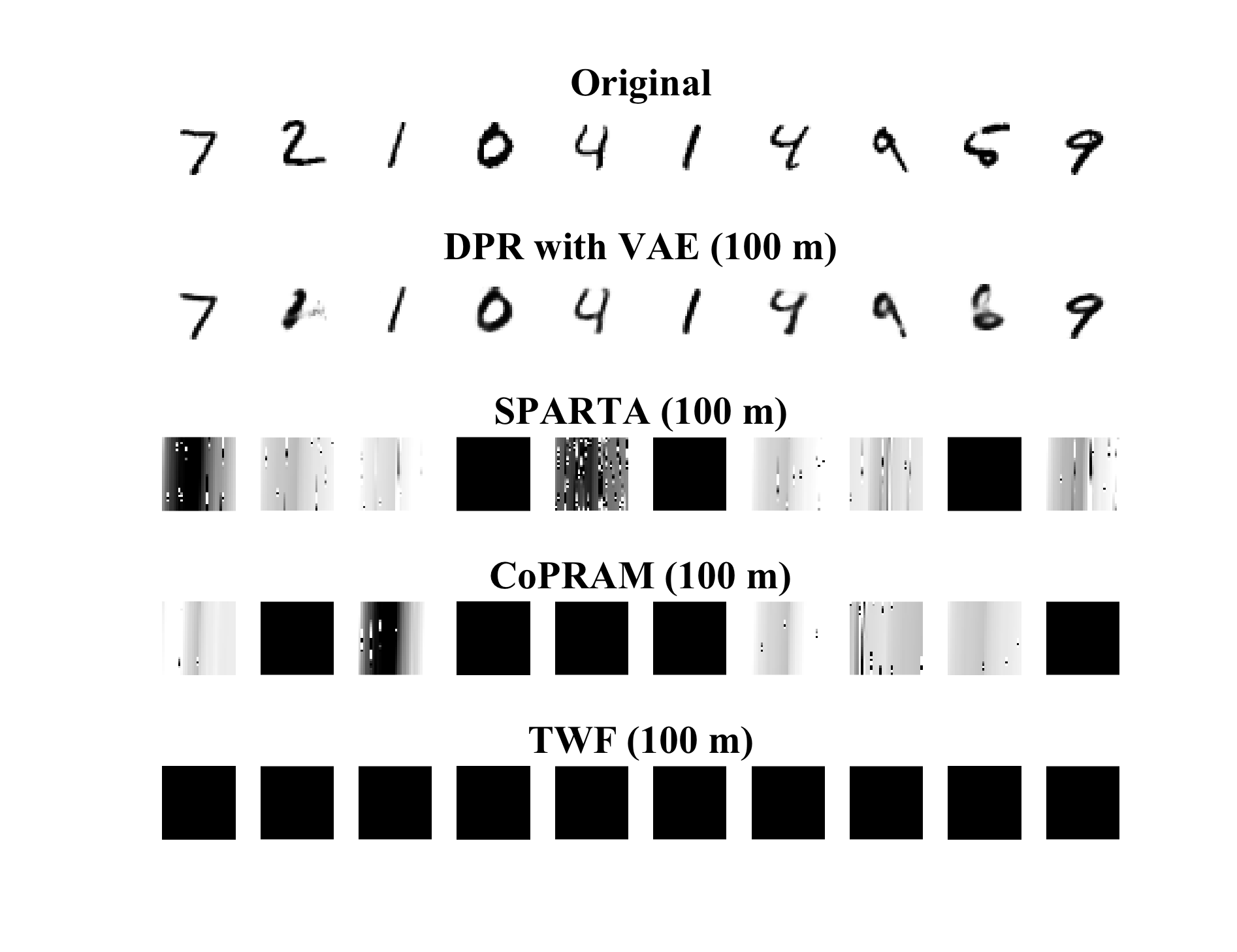}
  \includegraphics[scale = 0.1]{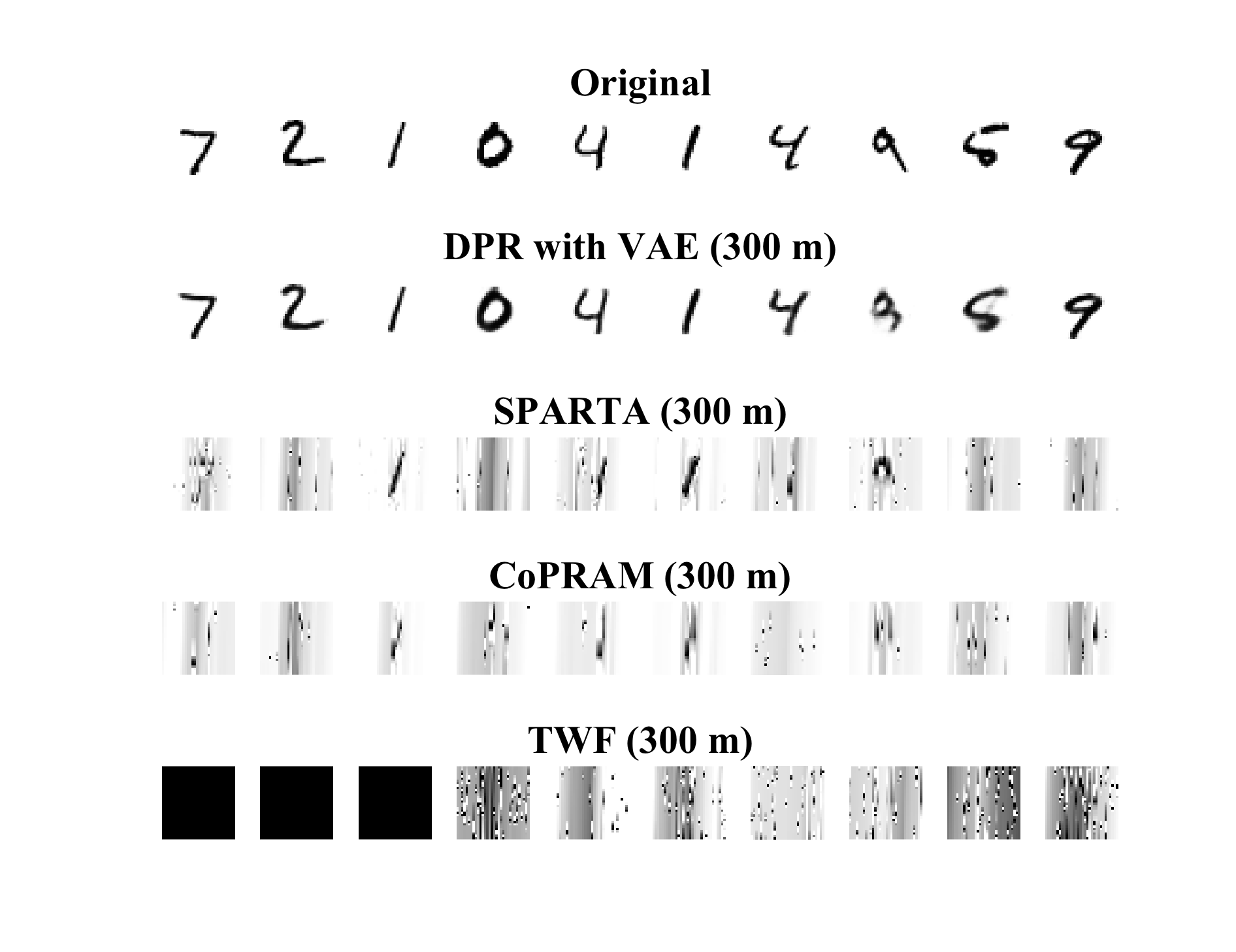}
  \caption{Each algorithm's reconstructed images with $100$ measurements (left) and $300$ measurements (right). If an image is blank, then the reconstruction error between the blank image and the original image was lower than that of the algorithm's reconstructed image and the original image. We note that even for as few as $100$ measurements, nearly all of DPR's reconstructions are semantically correct.}
  \label{fig:mnist_qualitative_results}
\end{figure}


In the second experiment, we used a pretrained Deep Convolutional Generative Adversarial Network (DCGAN) from \cite{Price2017} that was trained on the CelebA dataset \cite{CelebA}. This dataset consists of $200,000$ facial images of celebrities. The RGB images were cropped to be of size $64 \times 64$, resulting in vectorized images of dimension $64 \times 64 \times 3 = 12288.$ The latent code space dimension is $k = 100$. We allowed $2$ random restarts. We ran numerical experiments with the other methods and they did not succeed at measurement levels below $5000$. Figure \ref{fig:celebA_results_500} showcases our results in reconstructing $10$ images from the DCGAN's test set with $500$ measurements. Even for a limited number of measurements, our algorithm's reconstructions are faithful representations of the original images.

\vspace{-0.7mm}

\begin{figure}[!htb]
  \centering
  \includegraphics[scale = 0.35]{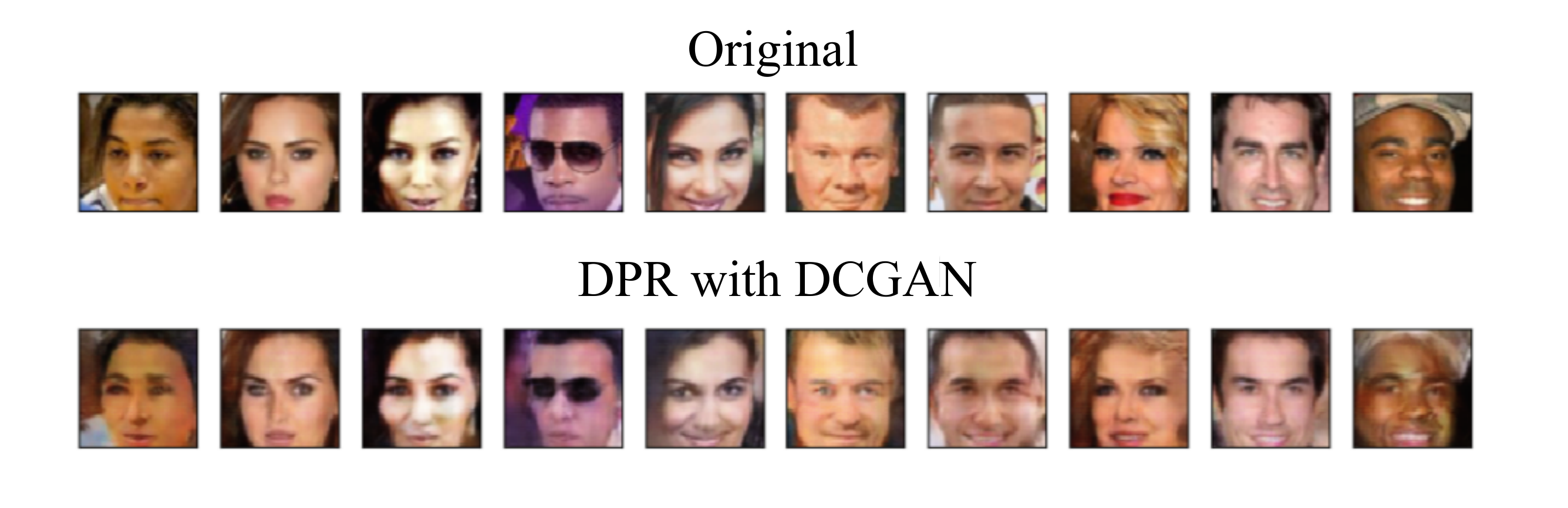}
  \caption{$10$ reconstructed images from celebA's test set using DPR with $500$ measurements.}
  \label{fig:celebA_results_500}
\end{figure}

\newpage

\bibliographystyle{plain}

\bibliography{main.bib}

\newpage

\section{Appendix}

\subsection{Full Proof of Theorem 3}

\begin{proof}
Set \begin{align*}
v_{x,x_0} = \begin{cases}
\nabla f(x)\ &\ f\ \text{is differentiable at}\ x \in \R^k\\
\lim_{\delta \rightarrow 0^+} \nabla f(x + \delta w)\ &\ \text{otherwise,}
\end{cases}
\end{align*} where $f$ is differentiable at $x + \delta w$ for sufficiently small $\delta > 0$. Recall that \begin{align*}
\nabla f(x) = (\PiWdix)^\top A^\top A (\PiWdix)x - (\PiWdix)^\top A_{G(x)}^\top A_{G(x_0)}(\PiWdixo)x_0.
\end{align*} Let \begin{align}
\overline{v}_{x,x_0} & := (\PiWdix)^\top(\PiWdix)x - (\PiWdix)^\top \Phi_{G(x),G(x_0)} (\PiWdixo)x_0, \label{vbar_def} \\
h_{x,x_0} & : = - \frac{\|x_0\|}{2^d}\left(\frac{\pi - 2\overline{\theta}_d}{\pi}\right)\left(\prod_{i=0}^{d-1}\frac{\pi - \overline{\theta}_i}{\pi}\right)\hat{x}_0 \\
& + \frac{1}{2^d}\left[\|x\| - \|x_0\|\left( \frac{2\sin \overline{\theta}_d}{\pi} - \left(\frac{\pi - 2\overline{\theta}_d}{\pi}\right) \sum_{i=0}^{d-1} \frac{\sin \overline{\theta}_i}{\pi}\left(\prod_{j=i+1}^{d-1} \frac{\pi - \overline{\theta}_i}{\pi}\right)\right)\right]\hat{x}, \label{h_xdef}
\end{align} and \begin{align*}
S_{\epsilon,x_0} & : = \left\{0 \neq x \in \R^k\ |\ \|h_{x,x_0}\| \leq \frac{1}{2^d}\epsilon \max(\|x\|,\|x_0\|)\right\}.
\end{align*} First, observe that by the WDC, we have that for all $x \neq 0$ and $i = 1,\dots,d$: \begin{align}
\left\|W_{i,+,x}^\top W_{i,+,x} - \frac{1}{2}I_{n_i}\right\| \leq \epsilon \Longrightarrow \|W_{i,+,x}\|^2 \leq \frac{1}{2} + \epsilon. \label{WDC_bound}
\end{align} Observe that \begin{align*}
\left\|\nabla f(x) - \overline{v}_{x,x_0}\right\| & \leq \left\|(\PiWdix)^\top(A^\top A - I_{n_d})(\PiWdix)x\right\| \\
& + \left\|(\PiWdix)^\top(A_{G(x)}^\top A_{G(x_0)} - \Phi_{G(x),G(x_0)})(\PiWdixo)x_0\right\|.
\end{align*} Hence by the RRCP and (\ref{WDC_bound}), we have that \begin{align}
\left\|\nabla f(x) - \overline{v}_{x,x_0}\right\| & \leq 14 \epsilon\left(\frac{1}{2} + \epsilon\right)^d \max(\|x\|,\|x_0\|) \label{RRCP_bound}.\end{align} 

Then Lemma 2 guarantees that for all non-zero $x,x_0 \in \R^k$: \begin{align}
\|\overline{v}_{x,x_0} - h_{x,x_0} \| \leq 78\frac{d^3}{2^d} \sqrt{\epsilon} \max(\|x\|,\|x_0\|). \label{h_bound}
\end{align} Then we have that for all non-zero $x,x_0 \in \R^k$: \begin{align*}
\|v_{x,x_0} - h_{x,x_0}\| & = \lim_{t \rightarrow 0^+} \|\nabla f(x + tw) - h_{x + t w,x_0}\| \\
& \leq \lim_{t \rightarrow 0^+} \left(\|\nabla f(x + tw) - \overline{v}_{x+tw,x_0}\| + \|\overline{v}_{x+tw,x_0} - h_{x + t w,x_0}\|\right) \\
& \leq \sqrt{\epsilon}\left(14 \frac{(1 + 2\epsilon)^d}{2^d} + 78\frac{d^3}{2^d}\right)\max(\|x\|,\|x_0\|) \\
& \leq \sqrt{\epsilon}K \frac{d^3}{2^d}\max(\|x\|,\|x_0\|)
\end{align*} for some universal constant $K$ where the first equality follows by the definition of $v_{x,x_0}$ and the continuity of $h_{x,x_0}$ for non-zero $x,x_0$. The second inequality combines (\ref{RRCP_bound}) and (\ref{h_bound}) and since $2\epsilon d \leq 1 \Longrightarrow (1 + 2\epsilon)^d \leq e^{2\epsilon d} \leq 1 + 4\epsilon d.$ This establishes concentration of $v_{x,x_0}$ to $h_{x,x_0}$ for all non-zero $x,x_0 \in \R^k$: \begin{align}
\|v_{x,x_0} - h_{x,x_0}\| \leq \sqrt{\epsilon}K \frac{d^3}{2^d}\max(\|x\|,\|x_0\|) \label{conc_v_h}
\end{align}

Now, due to the continuity and piecewise linearity of the function $G(x)$ and $|\cdot|$, we have that for any $x,y \neq 0$ that there exists a sequence $\{x_n\} \rightarrow x$ such that $f$ is differentiable at each $x_n$ and $D_yf(x) = \lim_{n \rightarrow \infty} \nabla f(x_n)\cdot y$. Thus, as $\nabla f(x_n) = v_{x_n,x_0}$, \begin{align*}
D_{-v_{x,x_0}}f(x) = - \lim_{n \rightarrow \infty} v_{x_n,x_0} \cdot v_{x,x_0}. \label{direc_deriv_def}
\end{align*} Then observe that \begin{align*}
v_{x_n,x_0} \cdot v_{x,x_0}  & =  h_{x_n,x_0} \cdot h_{x,x_0} + (v_{x_n,x_0} - h_{x_n,x_0})\cdot h_{x,x_0} + h_{x_n,x_0}\cdot (v_{x,x_0 } - h_{x,x_0}) \\
& + (v_{x_n,x_0} - h_{x_n,x_0})\cdot (v_{x,x_0} - h_{x,x_0}) \\
&  \geq h_{x_n,x_0} \cdot h_{x,x_0} - \|v_{x_n,x_0} - h_{x_n,x_0}\|\|h_{x,x_0}\| - \|h_{x_n,x_0}\|\|v_{x,x_0 } - h_{x,x_0}\| \\
& - \|v_{x_n,x_0} - h_{x_n,x_0}\|\|v_{x,x_0} - h_{x,x_0}\| \\
& \geq h_{x_n,x_0} \cdot h_{x,x_0}  - \|h_{x,x_0}\|\sqrt{\epsilon}K \frac{d^3}{2^d}\max(\|x\|,\|x_0\|)  \\
& - \|h_{x_n,x_0}\|\sqrt{\epsilon}K \frac{d^3}{2^d}\max(\|x\|,\|x_0\|) - \epsilon \left[K \frac{d^3}{2^d}\right]^2 \max(\|x_n\|,\|x_0\|)\max(\|x\|,\|x_0\|)
\end{align*} where in the last inequality, we used (\ref{conc_v_h}).  By the continuity of $h_{x,x_0}$ for non-zero $x \in \R^k$, we have that for $x \in S_{4\sqrt{\epsilon}Kd^3,x_0}^c$: \begin{align*}
\lim_{n \rightarrow \infty} v_{x_n,x_0} \cdot v_{x,x_0} & \geq \|h_{x,x_0}\|^2 - 2\|h_{x,x_0}\|\sqrt{\epsilon}K \frac{d^3}{2^d}\max(\|x\|,\|x_0\|) - \epsilon \left[K \frac{d^3}{2^d}\right]^2 \max(\|x\|,\|x_0\|)^2 \\
& = \frac{\|h_{x,x_0}\|}{2}\left(\|h_{x,x_0}\| - 4\sqrt{\epsilon}K \frac{d^3}{2^d}\max(\|x\|,\|x_0\|)\right) \\
& + \frac{1}{2}\left(\|h_{x,x_0}\|^2 - 2\epsilon \left[K \frac{d^3}{2^d}\right]^2\max(\|x\|,\|x_0\|)^2\right) \\
& > 0.
\end{align*} Hence we conclude that for all $x \in S_{4\sqrt{\epsilon}Kd^3,x_0}^c$, $D_{-v_{x,x_0}}f(x) < 0$.

We now show that $D_{x} f(0) < 0$ for all $x \neq 0$. Note that by direct calculation, \begin{align*}
D_xf(0) = - \frac{1}{2} \sum_{k=1}^m |\langle a_k, (\PiWdix)x\rangle \langle a_k, (\PiWdixo)x_0 \rangle|
\end{align*} where each $a_k$ denotes a row of $A$. We further note \begin{align*}
\left|\sum_{k=1}^m \langle a_k, (\PiWdix)x\rangle \langle a_k, (\PiWdixo)x_0 \rangle\right| = \left|\langle x,(\PiWdix)^\top A^\top A (\PiWdixo)x_0\rangle\right|
\end{align*} and thus \begin{align*}
D_xf(0) & \leq -\frac{1}{2}\left|\langle x,(\PiWdix)^\top A^\top A (\PiWdixo)x_0\rangle\right|  \\
& \leq - \frac{1}{2} \langle x,(\PiWdix)^\top A^\top A (\PiWdixo)x_0\rangle \\
& = - \frac{1}{2} \langle x,(\PiWdix)^\top( A^\top A  - I_{n_d})(\PiWdixo)x_0\rangle \\
& - \frac{1}{2} \langle x,(\PiWdix)^\top (\PiWdixo)x_0\rangle \\
& \stackrel{(*)}{\leq} \frac{1}{2}\left(\frac{7\epsilon (1/2+ \epsilon)^d}{2^d}\|x\|\|x_0\| - \frac{1}{4\pi} \frac{1}{2^d}\|x\|\|x_0\|\right) \\
& \leq \frac{1}{2}\left(\frac{14\epsilon }{2^d}\|x\|\|x_0\| - \frac{1}{4\pi} \frac{1}{2^d}\|x\|\|x_0\|\right) \\
& < 0
\end{align*} where $(*)$ follows by the RRCP and Lemma 3 as long as $4\epsilon d \leq 1$ and $\epsilon < 1/(56\pi)$.

We conclude by applying Proposition 1 and $24\pi d^6\sqrt{4\sqrt{\epsilon}Kd^3} \leq 1$ to attain \begin{align*}
S_{4\sqrt{\epsilon}Kd^6} \subset \mathcal{B}(x_0,89d\sqrt{4\sqrt{\epsilon}Kd^3}\|x_0\|) \cup \mathcal{B}(\rho_d x_0, 836831d^{12}\sqrt{4\sqrt{\epsilon}Kd^3}\|x_0\|).
\end{align*}
\end{proof}

We record some results that were used in the above proof. In \cite{Hand2017}, it was shown that Gaussian $W_i$ satisfies the WDC with high probability: 

\begin{lem}[Lemma 9 in \cite{Hand2017}]
Fix $0 < \epsilon < 1$. Let $W \in \R^{n \times k}$ have i.i.d. $\mathcal{N}(0,1/n)$ entries. If $n \geq c k \log k$ then with probability at least $1-8n\exp(-\gamma k)$, $W$ satisfies the WDC with constant $\epsilon$. Here $c,\gamma^{-1}$ are constants that depend only polynomially on $\epsilon^{-1}$.
\end{lem} The following is a technical result showing concentration of $\overline{v}_{x,x_0}$ around $h_{x,x_0}$: 

\begin{lem}
Fix $0 < \epsilon < d^{-4}(1/16\pi)^2$ and let $d \geq 2$. Let $W_i$ satisfy the WDC with constant $\epsilon$ for $i = 1,\dots d$. For any non-zero $x,y \in \R^k$, we have \begin{align*}
\|\overline{v}_{x,y} - h_{x,y}\| \leq \frac{78d^3}{2^d}\sqrt{\epsilon}\max(\|x\|,\|y\|).
\end{align*}
\end{lem}
\begin{proof} Observe that \begin{align*}
\|\overline{v}_{x,y} - h_{x,y}\| & \leq \underbrace{\left\|(\PiWdix)^\top(\PiWdix)x - \frac{1}{2^d}x\right\|}_{= Q_1}\  \\
& + \underbrace{\left\|\frac{\pi - 2\theta_d}{\pi}(\PiWdix)^\top(\PiWdiy)y - \frac{\pi - 2\overline{\theta}_d}{\pi}\tilde{h}_{x,y}\right\|}_{=Q_2}\ \\
& + \underbrace{\left\|\frac{2\sin \theta_d}{\pi}\frac{\|(\PiWdiy)y\|}{\|(\PiWdix)x\|}(\PiWdix)^\top(\PiWdix)x - \frac{2\sin \overline{\theta}_d}{\pi} \frac{\|y\|}{\|x\|}\frac{1}{2^d}x\right\|}_{=Q_3}.
\end{align*}

We focus on bounding each individual quantity $Q_i$ for $i = 1,2,3$:

$Q_1$: We have that by (\ref{conc_htilde}) in Lemma 3, \begin{align}
\left\|(\PiWdix)^\top(\PiWdix)x - \frac{1}{2^d}x\right\| \leq 24 \frac{d^3 \sqrt{\epsilon}}{2^d}\|x\|.
\end{align}

$Q_2$: We write \begin{align*}
Q_2 & \leq \left\|\frac{\pi - 2\theta_d}{\pi}(\PiWdix)^\top(\PiWdiy)y - \frac{\pi - 2\theta_d}{\pi}\tilde{h}_{x,y}\right\| \\
& + \left\|\frac{\pi - 2\theta_d}{\pi}\tilde{h}_{x,y} - \frac{\pi - 2\overline{\theta}_d}{\pi}\tilde{h}_{x,y}\right\| \\
& \stackrel{(*)}{\leq} \left|\frac{\pi - 2\theta_d}{\pi}\right|24 \frac{d^3\sqrt{\epsilon}}{2^d}\|y\| + \left|\frac{2}{\pi}(\theta_d - \overline{\theta}_d)\right|\|\tilde{h}_{x,y}\| \\
& \stackrel{(**)}{\leq} 24 \frac{d^3\sqrt{\epsilon}}{2^d}\|y\| + \frac{8d\sqrt{\epsilon}}{\pi}\frac{\left(1 + \frac{d}{\pi}\right)}{2^d}\|y\|
\end{align*} where in $(*)$ we used (\ref{conc_htilde}) and $(**)$ used (\ref{conc_angle}) and the fact that $\|\tilde{h}_{x,y}\| \leq 2^{-d}(1 + \frac{d}{\pi})\|y\|$. Hence \begin{align*}
Q_2 \leq \frac{1}{2^d} \left(24 d^3 + \frac{8d}{\pi}\left( 1+ \frac{d}{\pi}\right)\right)\sqrt{\epsilon}\|y\|.
\end{align*}

$Q_3$: Let $y_d : = (\PiWdiy)y$ and $x_d : = (\PiWdix)x$. Then we have that \begin{align*}
Q_3 & \leq \underbrace{\left|\frac{2\sin \theta_d}{\pi} - \frac{2\sin \overline{\theta}_d}{\pi}\right|\frac{\|y_d\|}{\|x_d\|}\left\|(\PiWdix)^\top x_d\right\|}_{=Q_{3,1}} \\
& + \underbrace{\left\|\frac{2\sin \overline{\theta}_d}{\pi}\frac{\|y_d\|}{\|x_d\|}(\PiWdix)^\top(\PiWdix)x - \frac{2\sin \overline{\theta}_d}{\pi} \frac{\|y\|}{\|x\|}(\PiWdix)^\top(\PiWdix)x \right\|}_{=Q_{3,2}} \\
& + \underbrace{\left\|\frac{2\sin \overline{\theta}_d}{\pi}\frac{\|y\|}{\|x\|}(\PiWdix)^\top(\PiWdix)x - \frac{2\sin \overline{\theta}_d}{\pi}\frac{\|y\|}{\|x\|}\frac{1}{2^d}x\right\|}_{=Q_{3,3}}.
\end{align*} Using (\ref{WDC_bound}) and (\ref{conc_angle}) gives \begin{align*}
Q_{3,1} & \leq \frac{2}{\pi}|\theta_d - \overline{\theta}_d| \left(\frac{1}{2} + \epsilon\right)^d\frac{\|y\|}{\|x\|}\|x\| \\
& \leq \frac{8d}{\pi}\left(\frac{1}{2} + \epsilon\right)^d \sqrt{\epsilon} \|y\| \\
& = \frac{8d(1 + 2\epsilon)^d}{\pi 2^d}\sqrt{\epsilon}\|y\|.
\end{align*} Likewise, equations (\ref{WDC_bound}) and (\ref{conc_norm}) gives \begin{align*}
Q_{3,2} & \leq \left|\frac{\|y_d\|}{\|x_d\|} - \frac{\|y\|}{\|x\|}\right|\left|\frac{2\sin \overline{\theta}_d}{\pi}\right|\left(\frac{1}{2} + \epsilon\right)^d\|x\| \\
& \leq 8d\epsilon \frac{\|y\|}{\|x\|} \frac{2}{\pi}\left(\frac{1}{2} + \epsilon\right)^d \|x\| \\
& \leq \frac{16d\sqrt{\epsilon}}{\pi} \left(\frac{1}{2} + \epsilon\right)^d \|y\| \\
& = \frac{16d(1+2\epsilon)^d}{\pi 2^d}\sqrt{\epsilon}\|y\|
\end{align*} Lastly, we use (\ref{conc_htilde}) to attain \begin{align*}
(3) & \leq \frac{2}{\pi} \frac{\|y\|}{\|x\|} 24\frac{d^3\sqrt{\epsilon}}{2^d}\|x\| \\
& \leq \frac{48d^3\sqrt{\epsilon}}{\pi2^d}\|y\|.
\end{align*} Combining the bounds for $Q_{3,i}$ for $i = 1,2,3$ gives \begin{align*}
Q_{3} & \leq Q_{3,1} + Q_{3,2} + Q_{3,3} \\ 
& \leq \frac{8d(1 + 2\epsilon)^d}{\pi 2^d}\sqrt{\epsilon}\|y\|+ \frac{16d(1+2\epsilon)^d}{\pi 2^d}\sqrt{\epsilon}\|y\|+ \frac{48d^3\sqrt{\epsilon}}{\pi2^d}\|y\| \\
& = \frac{1}{2^d}\left(\frac{24d(1 + 2\epsilon)^d + 48d^3}{\pi}\right)\sqrt{\epsilon}\|y\|.
\end{align*} Thus we attain \begin{align*}
Q_1 + Q_2 + Q_3 & \leq \frac{K_d}{2^d} \sqrt{\epsilon}\max(\|x\|,\|y\|)
\end{align*} where \begin{align*}
K_d & = 24d^3 + 24d^3 + \frac{8d(1 + d/\pi)}{\pi} + \frac{24(1 + 2\epsilon)^d}{\pi} + \frac{48d^3}{\pi} \\
& = \left( 48 + \frac{48}{\pi}\right)d^3 + \frac{8d}{\pi}\left(1 + \frac{d}{\pi} + 3d(1 + 2\epsilon)^d\right) \\
& \stackrel{(*)}{\leq} 64d^3 + \frac{8d}{\pi}\left( 1+ \frac{d}{\pi} + 3d( 1+ 4\epsilon d)\right) \\
& \leq 64d^3 + \frac{40d^3 + 96\epsilon d^3}{\pi} \\
& \stackrel{(**)}{\leq} 64d^3 + \frac{41d^3}{\pi} \\
& \leq 78d^3
\end{align*} where $(*)$ and $(**)$ hold as long as $\epsilon \leq \min(1/2d,1/96)$.
\end{proof}

The following result summarizes some useful bounds from \cite{Hand2017}:

\begin{lem}[Results from Lemma 5 in \cite{Hand2017}]
Fix $0 < \epsilon < d^{-4}(1/16\pi)^2$ and let $d \geq 2$. Let $W_i$ satisfy the WDC with constant $\epsilon$ for $i = 1,\dots d$. Then for any non-zero $x,y \in \R^k$, the following hold: \begin{align}
\left\|(\PiWdix)^\top(\PiWdiy)y - \tilde{h}_{x,y}\right\| & \leq 24 \frac{d^3\sqrt{\epsilon}}{2^d}\|y\|, \label{conc_htilde}\\
\left\langle (\PiWdix)x, (\PiWdiy)y \right\rangle & \geq \frac{1}{4\pi}\frac{1}{2^d}\|x\|\|y\|, \label{conc_inn_prod}\\
\left| \frac{\|y_d\|}{\|x_d\|} - \frac{\|y\|}{\|x\|} \right| & \leq 8d\epsilon \frac{\|y\|}{\|x\|}, \label{conc_norm}\\
|\theta_d - \overline{\theta}_d| & \leq 4d\sqrt{\epsilon} \label{conc_angle}
\end{align} where $x_d : = (\Pi_{j=d}^1 W_{j,+x})x$, $y_d: =(\Pi_{j=d}^1 W_{j,+y})y$, $\theta_d : = \angle(x_d,y_d)$, $\overline{\theta}_d : = g^{\circ d}(\angle(x,y))$, and the vector $\tilde{h}_{x,y}$ is defined as \begin{align*}
\tilde{h}_{x,y} : = \frac{1}{2^d}\left[ \left(\prod_{i=0}^{d-1} \frac{\pi - \overline{\theta}_i}{\pi}\right)y + \sum_{i=0}^{d-1}\frac{\sin\overline{\theta}_i}{\pi} \left(\prod_{j=i+1}^{d-1} \frac{\pi - \overline{\theta}_i}{\pi}\right)\frac{\|y\|}{\|x\|}x\right].
\end{align*}.
\end{lem}

\subsection{Determining where $h_{x,x_0}$ vanishes}

Before proving Proposition 1, we outline how the vector $h_{x,x_0}$ was derived. Recall that \begin{align*}
\nabla f(x) = (\PiWdix)^\top A^\top A (\PiWdix)x - (\PiWdix)^\top A_{G(x)}^\top A_{G(x_0)}(\PiWdixo)x_0.
\end{align*} The concentration of the first term follows by the RRCP and Lemma 3: \begin{align*}
(\PiWdix)^\top A^\top A (\PiWdix)x \approx (\PiWdix)^\top (\PiWdix)x \approx \tilde{h}_{x,x} = \frac{1}{2^d}x.
\end{align*} For the second term, note that the RRCP gives
\begin{align*}
(\PiWdix)^\top A_{G(x)}^\top A_{G(x_0)}(\PiWdixo)x_0 \approx (\PiWdix)^\top \Phi_{G(x),G(x_0)}(\PiWdixo)x_0.
\end{align*} Letting $x_d = (\PiWdix)x = G(x)$ and $x_{0,d} = (\PiWdixo)x_0 = G(x_0)$, note that \begin{align*}
\Phi_{x_d,x_{0,d}} = \frac{\pi - 2\theta_d}{\pi}I + \frac{2\sin \theta_d}{\pi} M_{\hat{x}_d \leftrightarrow \hat{x}_{0,d}}
\end{align*} where $\theta_d = \angle(x_d,x_{0,d})$. By Lemma 5 in \cite{Hand2017}, this angle is well-defined and $\|x_d\|,\|x_{0,d}\| \neq 0$ as long as each $W_i$ satisfies the WDC. Finally, note that the definition of $M_{\hat{x} \leftrightarrow \hat{y}}$ gives \begin{align*}
M_{\hat{x}_d \leftrightarrow \hat{x}_{0,d}} x_{0,d} & = \|x_{0,d}\| M_{\hat{x}_d \leftrightarrow \hat{x}_{0,d}} \hat{x}_{0,d} = \|x_{0,d}\| \hat{x}_d = \frac{\|x_{0,d}\|}{\|x_d\|}x_d.
\end{align*} Thus we see that \begin{align*}
&(\PiWdix)^\top  \Phi_{x_d,x_{0,d}}(\PiWdixo)x_0 \\
& = \frac{\pi - 2\theta_d}{\pi}(\PiWdix)^\top(\PiWdixo)x_0 + \frac{2\sin \theta_d}{\pi} \frac{\|x_{0,d}\|}{\|x_d\|}(\PiWdix)^\top(\PiWdix)x \\
& \approx \frac{\pi - 2\overline{\theta}_d}{\pi} \tilde{h}_{x,x_0} + \frac{2\sin\overline{\theta}_d}{\pi} \frac{\|x_0\|}{\|x\|} \frac{1}{2^d}x
\end{align*} where $\overline{\theta}_d = g^{\circ d}(\angle(x,x_0))$ and the definition of $\tilde{h}_{x,x_0}$ is given in Lemma 3. The concentration of the angle $\theta_d$ and norm $\|x_{0,d}\|/\|x_d\|$ are given in Lemma 3. Combining the concentrations of the two terms in $\nabla f(x)$ gives $h_{x,x_0}$.

Now, we establish that the set of all $x$ such that $\|h_{x,x_0}\| \approx 0$, denoted by $S_{\epsilon,x_0}$, is contained in two neighborhoods centered at $x_0$ and a negative multiple $-\rho_d x_0$.

\begin{prop} Suppose $24\pi d^6\sqrt{\epsilon}\leq 1.$ Let \begin{align*}
S_{\epsilon,x_0} = \left\{0 \neq x \in \R^k\ |\ \|h_{x,x_0}\| \leq \frac{1}{2^d}\epsilon \max(\|x\|,\|x_0\|)\right\}
\end{align*} where $d \geq 2$ and let 
\begin{align*}
h_{x,x_0} & = - \frac{\|x_0\|}{2^d}\left(\frac{\pi - 2\overline{\theta}_d}{\pi}\right)\left(\prod_{i=0}^{d-1}\frac{\pi - \overline{\theta}_i}{\pi}\right)\hat{x}_0 \\
& + \frac{1}{2^d}\left[\|x\| - \|x_0\|\left( \frac{2\sin \overline{\theta}_d}{\pi} - \left(\frac{\pi - 2\overline{\theta}_d}{\pi}\right) \sum_{i=0}^{d-1} \frac{\sin \overline{\theta}_i}{\pi}\left(\prod_{j=i+1}^{d-1} \frac{\pi - \overline{\theta}_i}{\pi}\right)\right)\right]\hat{x}.
\end{align*} where $\overline{\theta}_0 = \angle(x,x_0)$ and $\overline{\theta}_i = g(\overline{\theta}_{i-1})$. Define \begin{align*}
\rho_d : = - \frac{2 \sin \breve{\theta}_d}{\pi} + \left(\frac{\pi - 2 \breve{\theta}_d}{\pi}\right) \sum_{i=0}^{d-1}\frac{\sin \breve{\theta}_i}{\pi} \left(\prod_{j=i+1}^{d-1} \frac{\pi - \breve{\theta}_i}{\pi}\right)
\end{align*} where $\breve{\theta}_0 = \pi$ and $\breve{\theta}_i = g(\breve{\theta}_{i-1})$. If $x \in S_{\epsilon,x_0}$, then either \begin{align*}
|\overline{\theta}_0| \leq 2\sqrt{\epsilon}\ \text{and}\ |\|x\| - \|x_0\|| \leq 29d\sqrt{\epsilon}\|x_0\|
\end{align*} or \begin{align*}
|\overline{\theta}_0 - \pi|\leq 24\pi^2d^4\sqrt{\epsilon}\ \text{and}\ |\|x\| - \rho_d\|x_0\|| \leq 3517d^8\sqrt{\epsilon}\|x_0\|.
\end{align*} In particular, we have 
\begin{align*}
S_{\epsilon,x_0} \subset \mathcal{B}(x_0,89d\sqrt{\epsilon}\|x_0\|) \cup \mathcal{B}(-\rho_d x_0, 836831d^{12}\sqrt{\epsilon}\|x_0\|).
\end{align*} Additionally, $\rho_d \rightarrow 1$ as $d \rightarrow \infty$.
\end{prop}

\begin{proof}
Without loss of generality, let $x_0 = e_1$ and $\|x_0\| = 1$ where $e_1$ is the first standard basis vector in $\R^{k}$. We also set $x = \|x\|\left(\cos \overline{\theta}_0 e_1 + \sin \overline{\theta}_0 e_2\right)$ where $\overline{\theta}_0 = \angle(x,x_0)$. Then \begin{align*}
h_{x,x_0} & = - \frac{1}{2^d}\left(\frac{\pi - 2\overline{\theta}_d}{\pi}\right)\left(\prod_{i=0}^{d-1}\frac{\pi - \overline{\theta}_i}{\pi}\right)\hat{x}_0 \\
& + \frac{1}{2^d}\left[\|x\| - \left( \frac{2\sin \overline{\theta}_d}{\pi} - \left(\frac{\pi - 2\overline{\theta}_d}{\pi}\right) \sum_{i=0}^{d-1} \frac{\sin \overline{\theta}_i}{\pi}\left(\prod_{j=i+1}^{d-1} \frac{\pi - \overline{\theta}_i}{\pi}\right)\right)\right]\hat{x}.
\end{align*} Set \begin{align*}
\beta = \left(\frac{\pi - 2\overline{\theta}_d}{\pi}\right)\left(\prod_{i=0}^{d-1}\frac{\pi - \overline{\theta}_i}{\pi}\right)\ \text{and}\ \al = \frac{2\sin \overline{\theta}_d}{\pi} - \left(\frac{\pi - 2\overline{\theta}_d}{\pi}\right) \sum_{i=0}^{d-1} \frac{\sin \overline{\theta}_i}{\pi}\left(\prod_{j=i+1}^{d-1} \frac{\pi - \overline{\theta}_i}{\pi}\right)
\end{align*} with $r = \|x\|$ and $M = \max(r,1)$. Note that we can write \begin{align*}
h_{x,x_0} = \frac{1}{2^d}\left(-\beta \hat{x}_0 + (r - \al)\hat{x}\right)
\end{align*} Then if $x \in S_{\epsilon,x_0}$, we have that \begin{align}
|-\beta + \cos \overline{\theta}_0 (r - \al)| & \leq \epsilon M \label{cosleqM} \\
|\sin \overline{\theta}_0(r - \al)| & \leq \epsilon M. \label{sinleqM}
\end{align} We now tabulate some useful bounds from Lemma 8 in \cite{Hand2017}: \begin{align}
\overline{\theta}_i & \in [0,\pi/2]\ \text{for}\ i \geq 1 \\
\overline{\theta}_i & \leq \overline{\theta}_{i-1}\ \text{for}\ i \geq 1 \\ 
\left|\prod_{i=0}^{d-1}\frac{\pi - \overline{\theta}_i}{\pi}\right| & \leq 1 \label{piprod_upperbound}\\
\prod_{i=0}^{d-1}\frac{\pi - \overline{\theta}_i}{\pi} & \geq \frac{\pi - \overline{\theta}_0}{\pi d^3} \label{piprod_lowerbound}\\
\left|\sum_{i=0}^{d-1} \frac{\sin \overline{\theta}_i}{\pi}\left(\prod_{j=i+1}^{d-1} \frac{\pi - \overline{\theta}_i}{\pi}\right)\right| & \leq \frac{d}{\pi}\sin \overline{\theta}_0 \label{sum_prod_upperbound}\\
\overline{\theta}_0 = \pi + O_1(\delta) & \Longrightarrow \overline{\theta}_i = \breve{\theta}_i + O_1(i\delta) \label{pi_Oone_bound}\\
\overline{\theta}_0 = \pi + O_1(\delta) & \Longrightarrow \left|\prod_{i=0}^{d-1}\frac{\pi - \overline{\theta}_i}{\pi}\right| \leq \frac{\delta}{\pi} \\
\left|\frac{\pi - 2\overline{\theta}_i}{\pi}\right| & \leq 1\ \forall\ i \geq 1 \label{pi_twotheta_upperbound} \\
\overline{\theta}_d & \leq \cos^{-1}\left(\frac{1}{\pi}\right)\ \forall\ d \geq 2. \label{cos_inv_bound}
\end{align}

To prove the Proposition, we first show that it is sufficient to only consider the small and large angle case. Then, we show that in the small and large angle case, $x \approx x_0$ and $x \approx -\rho_d x_0$, respectively. We begin by proving that $\max(\|x\|,\|x_0\|) \leq 6d$ for any $x \in S_{\epsilon,x_0}$.

\textbf{Bound on maximal norm in $S_{\epsilon,x_0}$:} It suffices to show that $r \leq 6d$. Suppose $r > 1$ since if $r \leq 1$, the result is immediate. Then either $|\sin \overline{\theta}_0| \geq 1/\sqrt{2}$ or $|\cos \overline{\theta}_0| \geq 1/\sqrt{2}$. If $|\sin \overline{\theta}_0| \geq 1/\sqrt{2}$ then (\ref{sinleqM}) gives \begin{align*}
|r - \al| \leq \sqrt{2}\epsilon r \Longrightarrow (1-\sqrt{2}\epsilon)r \leq |\al|.
\end{align*} But \begin{align*}
|\al| & \leq \frac{2}{\pi}|\sin \overline{\theta}_d| + \left|\left(\frac{\pi - 2\overline{\theta}_d}{\pi}\right) \sum_{i=0}^{d-1} \frac{\sin \overline{\theta}_i}{\pi}\left(\prod_{j=i+1}^{d-1} \frac{\pi - \overline{\theta}_i}{\pi}\right)\right| \\
& \leq 1 + \frac{d}{\pi}
\end{align*} where the second inequality used equations (\ref{sum_prod_upperbound}) and (\ref{pi_twotheta_upperbound}). Thus \begin{align*}
r \leq \frac{1+ \frac{d}{\pi}}{1 - \sqrt{2}\epsilon} \leq 2\left(1 + \frac{d}{\pi}\right) \leq 2 + d \leq 2d
\end{align*} provided $\epsilon < 1/4$ and $d \geq 2$. If $|\cos \overline{\theta}_0 | \geq 1/\sqrt{2}$, then (\ref{cosleqM}) gives \begin{align*}
|r - \al| \leq \sqrt{2}(\epsilon r + |\beta|) \Longrightarrow (1 - \sqrt{2}\epsilon)r \leq \sqrt{2}|\beta| + \al.
\end{align*} But by (\ref{piprod_upperbound}),  \begin{align*}
|\beta| = \left|\left(\frac{\pi - 2\overline{\theta}_d}{\pi}\right) \left(\prod_{i=0}^{d-1} \frac{\pi - \overline{\theta}_i}{\pi}\right)\right| \leq 1\ \text{since}\ \overline{\theta}_i  \in [0,\pi/2]\ \forall\ i\ \geq 1.
\end{align*} Hence if $\epsilon < 1/4$, \begin{align*}
r \leq \frac{\sqrt{2} + 2d}{1 - \sqrt{2}\epsilon} \leq 2\sqrt{2} + 4d \leq \sqrt{2}d + 4d \leq 6d.
\end{align*} Thus in any case, $r \leq 6d \Longrightarrow M \leq 6d$. 

We now show that it is sufficient to only consider the small angle case $\overline{\theta}_0 \approx 0$ and the large angle case $\overline{\theta}_0 \approx \pi$.

\textbf{Sufficiency:} We have two possible situations:
\begin{itemize}
\item $|r - \al| \geq \sqrt{\epsilon}M$: Then (\ref{sinleqM}) implies \begin{align*}
|\sin \overline{\theta}_0| \leq \sqrt{\epsilon} \Longrightarrow \overline{\theta}_0 = O_1(2\sqrt{\epsilon})\ \text{or}\ \pi + O_1(2\sqrt{\epsilon}).
\end{align*}
\item $|r - \al| \leq \sqrt{\epsilon}M:$ Then (\ref{cosleqM}) implies \begin{align*}
|\beta| \leq 2 \sqrt{\epsilon}M.
\end{align*} But note that by (\ref{piprod_lowerbound}), \begin{align*}
\beta = \left(\frac{\pi - 2\overline{\theta}_d}{\pi}\right)\left(\prod_{i=0}^{d-1}\frac{\pi - \overline{\theta}_i}{\pi}\right) \geq \frac{(\pi - 2 \overline{\theta}_d)(\pi - \overline{\theta}_0)}{d^3\pi^2}.
\end{align*} In addition, (\ref{cos_inv_bound}) implies \begin{align*}
|\pi - 2\overline{\theta}_d| \geq \left|\pi - 2\cos^{-1}\left(\frac{1}{\pi}\right)\right|  \geq \frac{1}{2}.
\end{align*} Thus \begin{align*}
|\beta| \geq \frac{|(\pi - 2\overline{\theta}_d)(\pi - \overline{\theta}_0)|}{d^3 \pi^2} \geq \frac{|\pi - \overline{\theta}_0|}{2d^3\pi^2}
\end{align*} which implies \begin{align*}
|\pi - \overline{\theta}_0| \leq 4d^3\pi^2\sqrt{\epsilon}M \leq 24d^4\pi^2 \sqrt{\epsilon}.
\end{align*} Thus $\overline{\theta}_0 = \pi + O_1(24d^4\pi^2 \sqrt{\epsilon})$.
\end{itemize}

Lastly, we show that in the small angle case, $x \approx x_0$, while in the large angle case, $x \approx -\rho_d x_0$.

\textbf{Small Angle Case:} Assume $\overline{\theta}_0 = O_1(2\sqrt{\epsilon})$. Note that since $\overline{\theta}_i \leq \overline{\theta}_0 \leq 2\sqrt{\epsilon}$ for each $i$, we have that \begin{align*}
\prod_{i=0}^{d-1} \frac{\pi - \overline{\theta}_i}{\pi} 
\geq \left(1 - \frac{2\sqrt{\epsilon}}{\pi} \right)^d = 1 + O_1 \left(\frac{4d\sqrt{\epsilon}}{\pi}\right)
\end{align*} provided $2d \sqrt{\epsilon} \leq 1/2$. Hence \begin{align*}
\beta  & = \left(\frac{\pi - 2\overline{\theta}_d}{\pi}\right) \left(\prod_{i=0}^{d-1} \frac{\pi - \overline{\theta}_i}{\pi}\right) \\
& \geq \left(1 + O_1\left(\frac{4\sqrt{\epsilon}}{\pi}\right)\right)\left(1 + O_1\left(\frac{4d\sqrt{\epsilon}}{\pi}\right)\right)\
\end{align*} where we used (\ref{pi_Oone_bound}) in the second inequality. In addition, $|\sin \overline{\theta}_d| \leq |\overline{\theta}_d| \leq 2\sqrt{\epsilon}$ and (\ref{sum_prod_upperbound}) imply that \begin{align*}
\left|\sum_{i=0}^{d-1} \frac{\sin \overline{\theta}_i}{\pi}\left(\prod_{j=i+1}^{d-1} \frac{\pi - \overline{\theta}_i}{\pi}\right)\right| \leq \frac{d}{\pi}|\sin \overline{\theta}_d| \leq d\sqrt{\epsilon}. 
\end{align*} Hence \begin{align*}
\al & = \frac{2\sin \overline{\theta}_d}{\pi} - \left(\frac{\pi - 2\overline{\theta}_d}{\pi}\right) \sum_{i=0}^{d-1} \frac{\sin \overline{\theta}_i}{\pi}\left(\prod_{j=i+1}^{d-1} \frac{\pi - \overline{\theta}_i}{\pi}\right) \\
& = O_1\left(\frac{4\sqrt{\epsilon}}{3\pi}\right) + \left(1 + O_1\left(\frac{4\sqrt{\epsilon}}{\pi}\right)\right)O_1(d\sqrt{\epsilon}) \\
& = O_1\left(\frac{4\sqrt{\epsilon}}{3\pi}\right) + O_1(d\sqrt{\epsilon}) + O_1\left(\frac{4d\epsilon}{\pi}\right) \\
& = O_1 \left(\frac{(4 + 3d\pi + 12d)\sqrt{\epsilon}}{3\pi}\right)
\end{align*} Thus since $|-\beta + \cos \overline{\theta}_0(r - \al)| \leq \epsilon M$ and $M \leq 6d$, we attain \begin{align*}
- \left(1 + O_1\left(\frac{4\sqrt{\epsilon}}{\pi}\right)\right)\left(1 + O_1\left(\frac{4d\sqrt{\epsilon}}{\pi}\right)\right) & + (1 + O_1(2\epsilon))\left(r + O_1 \left(\frac{(4 + 3d\pi + 12d)\sqrt{\epsilon}}{3\pi}\right)\right) \\
& = O_1(6d\epsilon).
\end{align*} Rearranging, this gives \begin{align*}
r-1 & = O_1\left(\frac{4d\sqrt{\epsilon}}{\pi} + \frac{4\sqrt{\epsilon}}{\pi} + \frac{16d\epsilon}{\pi} + (2\epsilon + 1)\frac{(4 + 3d\pi + 12d)\sqrt{\epsilon}}{3\pi}\right) + O_1(12d\epsilon) + O_1(6d\epsilon)  \\
& = O_1\left(\frac{(12d + 12 + 48d)\sqrt{\epsilon} + (2\epsilon + 1)(4 + 3\pi d + 12d)\sqrt{\epsilon}}{3\pi} + 18d\sqrt{\epsilon}\right)\\
& = O_1(29d \sqrt{\epsilon}).
\end{align*}

\textbf{Large Angle Case:} Assume $\overline{\theta}_0 = \pi + O_1(\delta)$ where $\delta := 24d^4\pi^2\sqrt{\epsilon}.$  We first prove that $\al$ is close to $\rho_d$ in the $O_1$ sense. Recall that $\overline{\theta}_d = \breve{\theta}_d + O_1(d\delta)$. Then by the mean value theorem: \begin{align*}
|\sin \overline{\theta}_d - \sin \breve{\theta}_d | \leq |\overline{\theta}_d - \breve{\theta}_d | \leq d \delta
\end{align*} so $\sin \overline{\theta}_d = \sin \breve{\theta}_d + O_1(d\delta)$. Let \begin{align*}
\Gamma_d : = \sum_{i=0}^{d-1} \frac{\sin \breve{\theta}_i}{\pi}\left(\prod_{j=i+1}^{d-1} \frac{\pi - \breve{\theta}_i}{\pi}\right).
\end{align*} In \cite{Hand2017}, it was shown that if $\frac{d^2\delta}{\pi} \leq 1$, then
\begin{align*}
\sum_{i=0}^{d-1} \frac{\sin \overline{\theta}_i}{\pi}\left(\prod_{j=i+1}^{d-1} \frac{\pi - \overline{\theta}_i}{\pi}\right) = \Gamma_d + O_1(3d^3\delta).
\end{align*} Also $|\Gamma_d| \leq d$.
Thus \begin{align*}
\al & = \frac{2\sin \overline{\theta}_d}{\pi} - \left(\frac{\pi - 2\overline{\theta}_d}{\pi}\right) \sum_{i=0}^{d-1} \frac{\sin \overline{\theta}_i}{\pi}\left(\prod_{j=i+1}^{d-1} \frac{\pi - \overline{\theta}_i}{\pi}\right) \\
& = \frac{2\sin \breve{\theta}_d}{\pi} + O_1\left(\frac{2d\delta}{\pi}\right) - \left(\frac{\pi - 2 \breve{\theta}_d}{\pi} + O_1 \left(\frac{2d\delta}{\pi}\right)\right)\left(\Gamma_d + O_1(3d^3 \delta)\right) \\
& = \rho_d + O_1\left(\frac{2d\delta}{\pi}\right) + \Gamma_d O_1\left(\frac{2d\delta}{\pi}\right) + \left(\frac{\pi - 2\breve{\theta}_d}{\pi}\right)O_1\left(3d^3\delta\right) + O_1 \left(\frac{6d^4\delta^2}{\pi}\right) \\
& = \rho_d + O_1\left(\frac{2d\delta}{\pi}\right) + O_1\left(\frac{2d^2\delta}{\pi}\right)+ O_1\left(3d^3\delta\right) + O_1\left(\frac{6d^4\delta^2}{\pi}\right) \\
& = \rho_d + O_1 \left(\left(\frac{4\delta}{\pi} + 3\delta + \frac{6\delta^2}{\pi}\right)d^4\right) \\
& = \rho_d + O_1(7d^4\delta).
\end{align*} We now prove $r$ is close to $\rho_d$ in the $O_1$ sense. Since $x \in S_{\epsilon,x_0}$, \begin{align*}
|-\beta + \cos \overline{\theta}_0(r - \al)|\leq \epsilon M.
\end{align*} Also note that $|\beta| \leq \delta/\pi$. Since $\cos \overline{\theta}_0 = 1 + O_1(\overline{\theta}_0^2/2)$, we have that \begin{align*}
O_1(\delta/\pi) + (1 + O_1(\delta^2/2))(r - \rho_d + O_1(7d^4\delta)) = O_1(\epsilon M).
\end{align*} Using $r \leq 6d$, $\rho_d \leq 2d$, and $\delta = 24d^4\pi^2\sqrt{\epsilon} \leq 1$, we get \begin{align*}
r - \rho_d & + O_1 \left(\frac{\delta^2}{2}\right)( r - \rho_d) + O_1(7d^4 \delta) + O_1\left(\frac{7d^4\delta^3}{2}\right) = O_1(\epsilon M) + O_1 \left(\frac{\delta}{\pi}\right)  \\
\Longrightarrow r - \rho_d & = O_1\left(4d\delta^2 + 7d^4\delta + \frac{7d^4\delta^3}{2} + 6d\epsilon + \frac{\delta}{\pi}\right) \\
& = O_1\left(6d\epsilon + \delta\left(4d + 7d^4 + \frac{7d^4}{2} + \frac{1}{\pi}\right)\right) \\
& = O_1\left(\left(6d + 24d^4\pi^2\left(4d + \frac{21d^4}{2} + \frac{1}{\pi}\right)\right)\sqrt{\epsilon}\right) \\
& = O_1(3517d^8\sqrt{\epsilon}).
\end{align*}

Finally, to complete the proof we use the inequality \begin{align*}
\|x - x_0\| \leq |\|x\| - \|x_0\|| + \left(\|x_0\| + |\|x\| - \|x_0\||\right)\overline{\theta}_0.
\end{align*} This inequality states that if a two dimensional point is known to be within $\Delta r$ of magnitude $r$ and an angle $\Delta \theta$ away from $0$, then it is at most a Euclidean distance of $\Delta r + (r + \Delta r)\Delta \theta)$ away from the point $(r,0)$ in polar coordinates. Hence we attain \begin{align*}
S_{\epsilon,x_0} \subset \mathcal{B}(x_0,89d\sqrt{\epsilon}) \cup \mathcal{B}(\rho_d x_0, 836831d^{12}\sqrt{\epsilon}).
\end{align*} The result that $\rho_d \rightarrow 1$ as $d \rightarrow \infty$ follows from the following facts: \begin{align*}
-\frac{2\sin \breve{\theta}_d}{\pi} \rightarrow 0\ \text{as}\ d \rightarrow \infty\ \text{since}\ \breve{\theta}_d \rightarrow 0\ \text{as}\ d\rightarrow \infty
\end{align*} and in \cite{Hand2017}, it was shown that \begin{align*}
\sum_{i=0}^{d-1}\frac{\sin \breve{\theta}_i}{\pi} \left(\prod_{j=i+1}^{d-1} \frac{\pi - \breve{\theta}_i}{\pi}\right) \rightarrow 1\ \text{as}\ d \rightarrow \infty.
\end{align*} Hence \begin{align*}\left(\frac{\pi - 2 \breve{\theta}_d}{\pi}\right) \sum_{i=0}^{d-1}\frac{\sin \breve{\theta}_i}{\pi} \left(\prod_{j=i+1}^{d-1} \frac{\pi - \breve{\theta}_i}{\pi}\right) \rightarrow 1\ \text{as}\ d \rightarrow \infty
\end{align*} so $\rho_d \rightarrow 1$ as $d \rightarrow \infty$.
\end{proof}

\subsection{Gaussian matrices satisfy RRCP}

We set out to prove the following:
\begin{prop}
Fix $0 < \epsilon < 1$. Let $A \in \R^{m \times n_d}$ have i.i.d. $\mathcal{N}(0,1/m)$ entries. If $m > \tilde{c}dk \log (n_1n_2\dots n_d)$, then with probability at least $1 - \gamma m^{4k} \exp(-\hat{C} m)$, $A$ satisfies the RRCP with constant $\epsilon$. Here $\gamma$ is a positive universal constant, $\hat{C}$ depends on $\epsilon$, and $\tilde{c}$ depends polynomially on $\epsilon^{-1}$
\end{prop} To show that Gaussian $A$ satisfies the RRCP, we first establish that for any \textit{fixed} $z,w \in \R^n$, the inner product $\langle A_z^\top A_wx,y\rangle$ concentrates around its expectation $\langle \Phi_{z,w}x,y\rangle$ for all $x$ and $y$ in a fixed $k$-dimensional subspace of $\R^n$. As we will see by the end of this section, this fixed $k$-dimensional subspace will represent the range of our generative model. We first prove a simple technical result:

\begin{prop} Fix $z,w \in \R^n$ and $0 < \epsilon < 1$. Let $T \subset \R^n$ be a $k$-dimensional subspace. If \begin{align}
\left|\langle  A_z x, A_w x\rangle - \langle \Phi_{z,w} x, x\rangle\right| \leq \epsilon \|x\|^2\ \forall\ x \in T \label{tech_hypothesis}
\end{align} then \begin{align*}
\left|\langle  A_z x, A_w y\rangle - \langle \Phi_{z,w} x, y\rangle\right| \leq 3\epsilon\|x\|\|y\|\ \forall\ x,y \in T.
\end{align*}
\end{prop}

\begin{proof} Fix $0 < \epsilon < 1$. Suppose (\ref{tech_hypothesis}) holds and fix $x,y \in T$. Without loss of generality, assume $x$ and $y$ are unit normed. We will use the shorthand notation $\Phi = \Phi_{z,w}$. Since $T$ is a subspace, $x - y \in T$ so by (\ref{tech_hypothesis}), \begin{align*}
\left|\langle  A_z (x - y), A_w(x-y)\rangle - \langle \Phi (x-y), x-y\rangle\right| \leq \epsilon\|x-y\|^2
\end{align*} or equivalently \begin{align}
\langle \Phi (x-y), x-y\rangle - \epsilon \|x-y\|^2 \leq \langle  A_z (x - y), A_w(x-y)\rangle  \leq \langle \Phi (x-y), x-y\rangle + \epsilon \|x-y\|^2. \label{expand_out}
\end{align} Note that \begin{align*}
\|x-y\|^2 = 2 - 2\langle x,y\rangle,
\end{align*} \begin{align*}
\langle \Phi (x-y), x-y\rangle = \langle \Phi x,x\rangle + \langle \Phi y,y\rangle - 2\langle \Phi x,y\rangle,
\end{align*} and \begin{align*}
\langle  A_z (x - y), A_w(x-y)\rangle = \langle A_z x, A_w x\rangle + \langle A_z y, A_w y\rangle - 2 \langle A_z x, A_w y\rangle
\end{align*} where we used the fact that $\Phi$ and $A_z^\top A_w$ are symmetric. Rearranging (\ref{expand_out}) yields \begin{align*}
2 \left(\langle \Phi x,y \rangle - \langle A_z x, A_w y \rangle\right) \leq \left(\langle \Phi x, x\rangle - \langle  A_z x , A_w x\rangle\right) + \left(\langle \Phi  y, y\rangle - \langle  A_z y, A_w y\rangle\right) + (2 - 2\langle x,y \rangle)\epsilon.
\end{align*} By assumption, the first two terms are bounded from above by $\epsilon$. Thus \begin{align*}
2 \left(\langle \Phi x,y \rangle - \langle A_z x, A_w y \rangle\right) & \leq 2 \epsilon + (2 - 2\langle x,y\rangle)\epsilon \\
& = 2(2 - \langle x,y \rangle)\epsilon \\
& \leq 6 \epsilon
\end{align*} so \begin{align*}
\langle \Phi x,y \rangle - \langle A_z x, A_w y \rangle \leq 3 \epsilon.
\end{align*} The lower bound is identical. Hence \begin{align*}
\left|\langle \Phi x,y \rangle - \langle A_z x, A_w y \rangle \right| \leq 3 \epsilon.
\end{align*}
\end{proof}

We now require a variation of the Restricted Isometry Property typically proven for Gaussian matrices. In our situation, the matrix $A_z^\top A_w$ concentrates around $\Phi_{z,w} \neq I_{n}$ for $z \neq w$, so we must prove a generalization which we call the \textit{Restricted Concentration Property} (RCP). First, recall that for any $z,w \in \R^n$, $\E[A_z^\top A_w] = \Phi_{z,w}.$ In addition, we have that for any $x \in \R^n$, \begin{align*}
\left|\langle A_z^\top A_w x,x \rangle - \langle \Phi_{z,w}x,x \rangle\right| = \frac{1}{m} \left|\sum_{i=1}^m Y_i \right|
\end{align*} where \begin{align*}
Y_i = X_i - \E[X_i]\ \text{and}\ X_i = \sign(\langle a_i,z \rangle \langle a_i, w\rangle) \langle a_i , x\rangle^2.
\end{align*} Here each $a_i$ denotes an unnormalized row of $A$ in which $a_i \sim \mathcal{N}(0,I_n)$. Hence $Y_i$ are independent, centered, subexponential random variables. Thus they satisfy the following large deviation inequality: \begin{lem}[Corollary 5.17 in \cite{Vershynin_notes}]
Let $Y_1,\dots,Y_m$ be independent, centered, subexponential random variables. Let $K = \max_{i \in[m]} \|Y_i\|_{\psi_1}$. Then for all $\epsilon > 0$, \begin{align*}
\Pro\left(\frac{1}{m} \left|\sum_{i=1}^m Y_i \right| \geq \epsilon \right) \leq 2 \exp\left[-c \min\left(\frac{\epsilon^2}{K^2}, \frac{\epsilon}{K}\right)m\right]
\end{align*} where $c > 0$ is an absolute constant.
\end{lem} Observe that for any $x \in S^{n-1}$, Lemma 4 guarantees that for any fixed $z,w \in \R^n$ and $\epsilon > 0$, \begin{align}
\Pro\left(|\langle A_z^\top A_w x,x \rangle - \langle \Phi_{z,w} x,x \rangle| \geq \epsilon\right) \leq 2\exp(-c_0(\epsilon) m) \label{conc_bound}
\end{align} where $c_0(\epsilon) = c \min(\epsilon^2/K^2,\epsilon/K).$ We are now equipped to proceed with the proof of the RCP.

\begin{prop}[Variant of Lemma 5.1 in \cite{Baraniuk_lemma}: RCP] Fix $0 < \epsilon < 1$ and $k < m$. Let $A \in \R^{m \times n}$ have i.i.d. $\mathcal{N}(0,1/m)$ entries and fix $z,w \in \R^n$. Let $T \subset \R^n$ be a $k$-dimensional subspace. Then with probability exceeding $1 - 2(42/\epsilon)^k\exp(-c_0(\epsilon/8)m)$, \begin{align}
|\langle A_z^\top A_w x, x \rangle - \langle \Phi_{z,w} x, x \rangle| \leq \epsilon \|x\|^2\ \forall\ x \in T \label{RCP_bound}
\end{align} and \begin{align}
|\langle A_z^\top A_w x, y \rangle - \langle \Phi_{z,w} x, y \rangle| \leq 3\epsilon \|x\|\|y\|\ \forall\ x,y \in T. \label{RCP_3bound}
\end{align} Furthermore, let $U = \bigcup_{i=1}^M U_i$ and $V = \bigcup_{j=1}^N V_j$ where $U_i$ and $V_j$ are subspaces of $\R^n$ of dimension at most $k$ for all $i$ and $j$. Then
\begin{align}
\left|\langle  A_{z}^\top A_w u, v \rangle - \langle \Phi_{z,w} u,v\rangle\right| \leq 3\epsilon \|u\|\|v\|\ \forall\ u \in U,\ v \in V \label{RCP_union_bound}
\end{align} with probability exceeding $1 - 2MN(42/\epsilon)^{2k}\exp(-c_0(\epsilon/8)m)$. Here $c_0(\epsilon/8) > 0$ only depends on $\epsilon$.

\end{prop}

\begin{proof}

Fix $0 < \epsilon < 1$ and $k < m$. Since $A$ is Gaussian, we may take $T$ to be in the span of the first $k$ standard basis vectors. In addition, assume $\|x\| = 1$ for any $x \in T$. For notational simplicity, set $\Sigma : = A_z^\top A_w - \Phi_{z,w}$.  Choose a finite set of points $Q_T \subset T$ each with unit norm such that $|Q_T| \leq (42/\epsilon)^k$ and for any $x \in T$, \begin{align}
\min_{q \in Q_T}\|x - q\| \leq \frac{\epsilon}{14}. \label{min_bound}
\end{align} Then we may apply a union bound to (\ref{conc_bound}) for this set of points to attain \begin{align}
\Pro\left(|\langle \Sigma q,q \rangle| \geq \frac{\epsilon}{8}\ \forall\ q \in Q_T\right) \leq 2 \left(\frac{42}{\epsilon}\right)^k\exp\left[-c_0\left(\frac{\epsilon}{8}\right)m\right].
\end{align} Now, define \begin{align}
\al^* : = \inf\left\{\al > 0 \ |\ |\langle \Sigma x,x\rangle| \leq \al\|x\|^2\ \forall\ x \in T\right\}. \label{alpha_star_def}
\end{align} We want to show that $\al^* \leq \epsilon$. Fix $x \in T$ with unity norm. Then there exists a $q \in Q_T$ with $\|q\| = 1$ such that $\|x - q \| \leq \epsilon/14.$ In addition, observe that $x - q \in T$ since $q \in Q_T \subset T$ so by (\ref{alpha_star_def}), \begin{align}
|\langle \Sigma(x-q),x-q\rangle| \leq \al^*\|x-q\|^2 \leq \al^*\frac{\epsilon^2}{196}.
\end{align} Now, note that by the definition of $\al^*$, \begin{align*}
|\langle \Sigma x, x \rangle| \leq \al^*\ \forall\ x \in T.
\end{align*} Thus Proposition 3 gives \begin{align*}
|\langle \Sigma x, y \rangle| \leq 3\al^*\ \forall\ x,y \in T.
\end{align*} Applying this result to $x-q$ and $q$ gives \begin{align*}
|\langle \Sigma (x-q), q \rangle| \leq 3\al^*\|x- q\| \leq \al^*\frac{3\epsilon}{14}.
\end{align*} Using $\langle \Sigma x,x \rangle = \langle \Sigma(x-q),x-q \rangle + 2\langle \Sigma x,q \rangle - \langle \Sigma q, q \rangle$ and $\langle \Sigma x,q \rangle = \langle \Sigma (x-q),q\rangle + \langle \Sigma q,q \rangle$, we see that \begin{align*}
|\langle \Sigma x,x \rangle| & \leq |\langle \Sigma(x-q),x-q\rangle| + 2|\langle \Sigma x,q\rangle| + |\langle \Sigma q,q\rangle| \\
& \leq |\langle \Sigma(x-q),x-q\rangle| + 2|\langle \Sigma (x-q),q\rangle| + 3|\langle \Sigma q,q\rangle| \\
& \leq \al^* \frac{\epsilon^2}{196} + \al^*\frac{3\epsilon}{7} + \frac{3\epsilon}{8} \\
& = \al^*\left(\frac{\epsilon^2}{196} + \frac{3\epsilon}{7}\right) + \frac{3\epsilon}{8}.
\end{align*} Note that this bound can be derived for any $x \in T$ because we can always find a $q \in Q_T$ with $\|q \| = 1$ such that $\|x - q\| \leq \epsilon/14.$ Thus \begin{align}
|\langle \Sigma x,x \rangle| \leq \al^*\left(\frac{\epsilon^2}{196} + \frac{3\epsilon}{7}\right) + \frac{3\epsilon}{8}\ \forall\ x \in T. \label{alpha_star_bound}
\end{align} However, recall that $\al^*$ was defined to be the smallest number such that \begin{align*}
|\langle \Sigma x,x \rangle| \leq \al^*\ \forall\ x \in T.
\end{align*} Hence $\al^*$ must be smaller than the right hand side of (\ref{alpha_star_bound}), i.e. \begin{align*}
\al^* \leq \al^*\left(\frac{\epsilon^2}{196} + \frac{3\epsilon}{7}\right) + \frac{3\epsilon}{8} \Longrightarrow \al^* \leq \frac{3\epsilon}{8}\left(\frac{1}{1 - \frac{\epsilon^2}{196} - \frac{3\epsilon}{7}}\right) \leq \epsilon
\end{align*} since $0 < \epsilon < 1$. Hence we conclude that with probability exceeding $1 - 2(42/\epsilon)^k\exp(-c_0(\epsilon/8)m)$, \begin{align*}
|\langle \Sigma x,x \rangle| \leq \epsilon \|x\|^2\ \forall\ x \in T
\end{align*} i.e. \begin{align*}
|\langle A_z^\top A_w x, x\rangle - \langle \Phi_{z,w} x,x \rangle| \leq \epsilon \|x\|^2\ \forall\ x \in T.
\end{align*} Applying Proposition 3 to our result gives (\ref{RCP_3bound}). The extension to the union of subspaces follows by applying (\ref{RCP_3bound}) to all subspaces of the form $\text{span}(U_i,V_j)$ and using a union bound.

\end{proof}

Now, this result establishes the concentration of $\langle A_z^\top A_w x,y \rangle$ around $\langle \Phi_{z,w}x,y \rangle$ for $x$ and $y$ in a fixed $k$-dimensional subspace for \textit{fixed} $z,w \in \R^n$. However, in reality, we are interested in showing that this concentration holds for all $z$ and $w$ in the range of our generative model. Hence we require the following extension of the RCP, which holds uniformly for all $z$ and $w$ in (possibly) different $k$-dimensional subspaces. The proof of this result uses an interesting fact from $1$-bit compressed sensing \cite{Vershynin_hyperplane_thm} which establishes that if two sparse vectors lie on the same side of a sufficient number of hyperplanes, then they are approximately equal with high probability.

\begin{prop}[Uniform RCP]
Fix $0 < \epsilon < 1$ and $k < m$. Let $A \in \R^{m \times n}$ have i.i.d. $\mathcal{N}(0,1/m)$ entries. Let $Z,W$, and $T$ be fixed $k$-dimensional subspaces of $\R^n$. If $m \geq k(c/\epsilon)^5$, then with probability at least $1 - \gamma (42/\epsilon)^k m^{4k}\exp(-C \epsilon m),$ we have \begin{align}
\left|\langle A_z^\top A_w x,y \rangle - \langle \Phi_{z,w}x,y \rangle\right| \leq 7\epsilon \|x\|\|y\|\ \forall\ x,y \in T,\ z \in Z,\ w \in W \label{Uniform_RCP_firstbound}
\end{align} where $\gamma$ and $c$ are positive universal constants and $C$ only depends on $\epsilon$. Furthermore, let $U = \bigcup_{i=1}^M U_i$ and $V = \bigcup_{j=1}^N V_j$ where $U_i$ and $V_j$ are subspaces of $\R^n$ of dimension at most $k$ for all $i$ and $j$. Then
\begin{align}
\left|\langle A_z^\top A_w u,v \rangle - \langle \Phi_{z,w}u,v \rangle\right| \leq 7\epsilon \|u\|\|v\|\ \forall\ u \in U,\ v \in V,\ z \in Z,\ w \in W \label{Uniform_RCP_secondbound}
\end{align} with probability exceeding $1 - MN\gamma (42/\epsilon)^{2k}m^{4k}\exp(-C \epsilon m)$.
\end{prop}

\begin{proof}
Fix $0 < \epsilon < 1$. Let $\{a_k\}_{k=1}^m$ denote the rows of $A$. Consider the sets \begin{align*}
Z_0 & : = \left\{z_i \in Z\ |\ \|z_i\| = 1\ \text{and}\ a_k^\top z_i \neq 0\ \forall\ k \in [m],\ i \in I\right\}, \\
W_0 & : = \left\{w_j \in W\ |\ \|w_j\| = 1\ \text{and}\ a_k^\top w_j \neq 0\ \forall\ k \in [m],\ j \in J\right\}
\end{align*} for some finite index sets $I$ and $J$. Then define the following event: \begin{align*}
E_{Z,A} : = \left\{|I| \leq 10m^{2k}\ \text{and}\ \forall\ z \in Z\ \text{s.t.}\ \|z\| = 1,\ \exists\ z_i \in Z_0\ 
\text{s.t.}\ \|z - z_i\| \leq \epsilon\right\}.
\end{align*} One can define the analogous event $E_{W,A}$ for $W_0$ and $J$.

By Lemma 4, we have that the unit sphere of $Z$ is partitioned into at most $10m^{2k}$ regions by the rows $\{a_k\}_{k=1}^m$ of $A$ with probability $1$. Then choose $\{z_i\}_{i \in I}$ as a set of representative points in the interior of each region. To use Theorem 4, we set $n = s = k$ to get \begin{align*}
\delta : = C_1 \left(\frac{k}{m} \log (2)\right)^{1/5} \Longrightarrow \delta = c \left(\frac{k}{m}\right)^{1/5}\ \text{where}\ c = C_1 \left(\log(2)\right)^{1/5}
\end{align*} for some positive universal constant $C_1$. Then we have that by Theorem 4, if $n = s = k$ and $m \geq k(c/\epsilon)^5$, then \begin{align*}
\delta = c \left(\frac{k}{m}\right)^{1/5} \leq ck^{1/5}\left(\frac{\epsilon^5}{kc^5}\right)^{1/5} = \epsilon
\end{align*} so  $\Pro_A(E_{Z,A}) \geq 1 - C_2 \exp(-\tilde{c}\epsilon m)$ for some positive universal constants $c,\tilde{c},$ and $C_2$. In addition, the event $E_{W,A}$ holds with the same probability.

Now, we apply Proposition 4 (RCP) and a union bound over all $\{z_i\}_{i \in I}$ and $\{w_j\}_{j \in J}$ to attain \begin{align}
\left|\langle A_{z_i}^\top A_{w_j} x,y \rangle - \langle \Phi_{z_i,w_j}x,y \rangle\right| & \leq 3\epsilon \|x\|\|y\|\ \forall\ x,y \in T,\ i \in I,\ j \in J \label{conc_bound_Zi_Wj}
\end{align} with probability at least $1 - 2C_2(42/\epsilon)^k|I||J|\exp(-C \epsilon m) \geq 1 - \gamma (42/\epsilon)^k m^{4k}\exp(-C \epsilon m)$ where $\gamma$ is a positive absolute constant and $C$ depends on $\epsilon$.

Finally, on the event $E_{Z,A} \cap E_{W,A}$, we have that for all $z \in Z$ and $w \in W$ with $\|z\| = \|w\| = 1$, there exists a $z_i \in Z_0$ and $w_j \in W_0$ such that \begin{align*}
\left|\langle A_z^\top A_w x,y \rangle - \langle \Phi_{z,w}x,y \rangle\right| & = \left|\langle A_{z_i}^\top A_{w_j} x,y \rangle - \langle \Phi_{z,w}x,y \rangle\right| \\
& \leq \left|\langle A_{z_i}^\top A_{w_j} x,y \rangle - \langle \Phi_{z_i,w_j}x,y \rangle\right| + \left|\langle \Phi_{z_i,w_j}x,y \rangle - \langle \Phi_{z,w}x,y \rangle\right| \\
& \leq 3\epsilon \|x\|\|y\| + 4\epsilon \|x\|\|y\| \\
& = 7\epsilon \|x\|\|y\|
\end{align*} where we used the continuity of $\Phi_{z,w}$ (Lemma 5) and (\ref{conc_bound_Zi_Wj}) in the second inequality. The extension to the union of subspaces follows by applying (\ref{Uniform_RCP_firstbound}) to all subspaces of the form $\text{span}(U_i,V_j)$ and using a union bound.
\end{proof}

The following results were used in the proof of Proposition 5:

\begin{lem}[Continuity of $\Phi_{z,w}$]
For any $z,w \in \R^n$, we have \begin{align*}
\|\Phi_{z + \delta z,w + \delta w} - \Phi_{z,w}\| \leq 2\left(\|\delta z\| + \|\delta w\|\right).
\end{align*}
\end{lem}

\begin{thm}[Theorem 2.1 in \cite{Vershynin_hyperplane_thm}]
Let $n,m,s > 0$ and set $\delta = C_1 \left(\frac{s}{m} \log(2n/s)\right)^{1/5}$. Let $a_i \in \R^n$ have i.i.d. $\mathcal{N}(0,1)$ entries for $i \in [m]$. Then with probability at least $1 - C_2 \exp(- c \delta m)$, the following holds uniformly for all $x,\tilde{x} \in \R^n$ that satisfy $\|x\|_2 = \|\tilde{x}\|_2 = 1$, $\|x\|_1 \leq \sqrt{s}$, and $\|\tilde{x}\|_1 \leq \sqrt{s}$ for $s \leq n$: \begin{align}
\langle a_i, \tilde{x} \rangle \langle a_i, x \rangle \geq 0,\ i \in [m] \Longrightarrow \|\tilde{x} - x\|_2 \leq \delta.
\end{align} Here $C_1, C_2, c$ are positive universal constants.
\end{thm}

\begin{lem}
Let $V$ be a subspace of $\R^n$. Let $A \in \R^{m \times n}$ have i.i.d. $\mathcal{N}(0,1/m)$ entries. With probability 1, \begin{align*}
|\{\diag(\sign(Av))A\ |\ v \in V\}| \leq 10m^{2\dimension V}.
\end{align*}
\end{lem}

\begin{proof}
It suffices to prove the same upperbound for $|\{\sgn(Av)\ |\ v \in V\}|.$ Let $\ell = \dim V$. By rotational invariance of Gaussians, we may take $V = \text{span}(e_1,\dots,e_{\ell})$ without loss of generality. Without loss of generality, we may let $A$ have dimensions $m \times 
\ell$ and take $V = \R^{\ell}$.

We will appeal to a classical result from sphere covering \cite{sphere_covering}. If $m$ hyperplanes in $\R^{\ell}$ contain the origin and are such that the normal vectors to any subset of $\ell$ of those hyperplanes are independent, then the complement of the union of these hyperplanes is partitioned into at most \begin{align*}
2 \sum_{i = 0}^{\ell - 1} \binom{m-1}{i}
\end{align*} disjoint regions. Each region uniquely corresponds to a constant value of $\sgn(Av)$ that has all non-zero entries. With probability $1$, any subset of $\ell$ rows of $A$ are linearly independent, and thus, \begin{align*}
|\{\sgn(Av)\ |\ v \in \R^{\ell},\ (Av)_i \neq 0\ \forall\ i\}| \leq 2 \sum_{i=0}^{\ell - 1} \binom{m-1}{i} \leq 2\ell \left(\frac{e m}{\ell}\right)^{\ell} \leq 10m^{\ell}
\end{align*} where the first inequality uses the fact that $\binom{m}{\ell} \leq (em/\ell)^{\ell}$ and the second inequality uses that $2\ell (e/\ell)^{\ell} \leq 10$ for all $\ell \geq 1$.

For arbitrary $v$, at most $\ell$ entries of $Av$ can be zero by linear independence of the rows of $A$. At each $v$, there exists a direction $\tilde{v}$ such that $(A(v + \delta \tilde{v}))_{i} \neq 0$ for all $i$ and for all $\delta$ sufficiently small. Hence, $\sgn(Av)$ differs from one of $\{\sgn(Av)\ |\ v \in \R^{\ell},\ (Av)_{i} \neq 0\ \forall\ i\}$ by at most $\ell$ entries. Thus, \begin{align*}
|\{\sgn(Av)\ |\ v \in \R^{\ell}\}| \leq \binom{m}{\ell} |\{\sgn(Av)\ |\ v \in \R^{\ell},\ (Av)_i \neq 0\ \forall\ i\}| \leq m^{\ell} 10 m^{\ell} = 10m^{2\ell}.
\end{align*}
\end{proof}

With the Uniform RCP, we may now prove the RRCP:

\begin{prop}[Range Restricted Concentration Property (RRCP)]
Fix $0 < \epsilon < 1$. Let $W_{i} \in \R^{n_i \times n_{i-1}}$ have i.i.d. $\mathcal{N}(0,1/n_i)$ entries for $i = 1,\dots,d$. Let $A \in \R^{m \times n_d}$ have i.i.d. $\mathcal{N}(0,1/m)$ entries. If $m > \tilde{c}dk \log (n_1n_2\dots n_d)$, then with probability at least $1 - \gamma m^{4k} \exp(-\hat{C} m)$, we have that for all $x,y \in \R^k$, \begin{align*}
%
\|(\PiWdix)^\top( A_{x_d}^\top A_{y_d} - \Phi_{x_d,y_d} )(\PiWdiy)\| \leq 7\epsilon \prod_{i=1}^d\|W_{i,+,x}\|\|W_{i,+,y}\|
\end{align*} where \begin{align*}
x_d : = (\PiWdix)x\ \text{and}\ y_d : = (\PiWdiy)y.
\end{align*} Here $\gamma$ is a positive universal constant, $\hat{C}$ depends on $\epsilon$, and $\tilde{c}$ depends polynomially on $\epsilon^{-1}$
\end{prop}

\begin{proof}
It suffices to show that for all $x,y,w,v \in S^{k-1}$, \begin{align}
\left|\langle (A_{x_d}^\top A_{y_d}-\Phi_{x_d,y_d})(\PiWdix)w,(\PiWdiy)v \rangle\right|\leq 7\epsilon \prod_{i=1}^d\|W_{i,+,x}\|\|W_{i,+,y}\|. \label{RRCP}
\end{align} We will use (\ref{Uniform_RCP_secondbound}) from Proposition 5. Note that $\dimension \left(\text{range}(\PiWdix)\right) \leq k$ for all $x \neq 0$. It has been shown in Lemma 15 of \cite{Hand2017} that \begin{align*}
|\{ (\PiWdix)x\ |\ x \neq 0\}| \leq 10^{d^2}(n_1^d n_2^{d-1}\dots n_{d-1}^2n_d)^k.
\end{align*} Hence it follows that $\{(\PiWdix) w\ |\ x,w \in S^{k-1}\} \subset U$ where $U$ is the union of at most $10^{d^2}(n_1^dn_2^{d-1}\dots n_{d-1}^2n_d)^k$ subspaces of dimensionality at most $k$. We can similarly conclude $\{(\PiWdiy) v\ |\ y,v \in S^{k-1}\} \subset V$ where $V$ is the union of at most $10^{d^2}(n_1^dn_2^{d-1}\dots n_{d-1}^2n_d)^k$ subspaces of dimensionality at most $k$. Hence applying (\ref{Uniform_RCP_secondbound}) from Proposition 5 with $Z = \text{range}(\PiWdix)$ and $W = \text{range}(\PiWdiy)$ gives (\ref{RRCP}) with probability at least \begin{align}
1- \gamma m^{4k} 10^{d^2} (42 n_1^dn_2^{d-1}\dots n_{d-1}^2n_d/\epsilon)^{2k} \exp(-C \epsilon m) \geq 1- \gamma m^{4k} \exp(-\hat{C} m)
\end{align} provided $m \geq \tilde{c}dk \log (n_1 n_2 \dots n_d)$ for some $\hat{C} = C\epsilon/2$ and $\tilde{c} = \Omega (\epsilon^{-1}\log \epsilon)$.
\end{proof}

\end{document}